\documentclass[12pt]{article}
\usepackage[utf8]{inputenc}
\usepackage{hyperref}
\usepackage{sgame}
\usepackage[english]{babel}
\usepackage{amsthm}
\usepackage[nottoc,numbib]{tocbibind}
\usepackage{amsmath,amssymb}
\usepackage{mathtools}
\usepackage{float}
\usepackage{titling}
\usepackage{blindtext}
\usepackage{bbm}
\usepackage{nccmath}
\usepackage[ruled, vlined]{algorithm2e}
\usepackage{graphicx}
\usepackage{caption}
\usepackage{subcaption}
\usepackage[export]{adjustbox}[2011/08/13]
\usepackage{listings}
\newtheorem{theorem}{Theorem}[section]

\newtheorem{example}[theorem]{Example}
\newtheorem{prop}[theorem]{Prop}
\newtheorem*{remark}{Remark}

\title{Strategies for the Iterated Prisoner's Dilemma}
    \author{Anagh Malik\thanks{am10118@ic.ac.uk}}
    \date{November 2020}

\begin{document}

\begin{titlingpage}
    \maketitle

    \begin{abstract}
    We explore some strategies which tend to perform well in the IPD. We start off by showing the significance of Tit-For-Tat strategies in evolutionary game theory. This is followed by a theoretical derivation of zero-determinant strategies, where we highlight an error on bounds for scale parameters from the original paper on ZD strategies\cite{press}. We then present examples of such strategies and create a custom player drawing inspiration from Markov Decision Processes.  At the end we pit them all against each other and see how they perform in an IPD tournament. Code can be found at \\ \href{https://github.com/anaghmalik/Axelrod}{https://github.com/anaghmalik/Axelrod}. 
    \end{abstract}
\end{titlingpage}

\tableofcontents

\newpage

\section{Note from the author}
I am especially proud of my Custom Strategy in \ref{sec:cust}, which draws inspiration from Reinforcement Learning and a unique interpretation of Tit-For-Tat strategies. I am also proud of the watertight explanation of the zero-determinant strategies, in which I heavily banked on Markov Processes and Linear Algebra to solidify the ideas from the paper \cite{press}. On page . 

\section{Introduction}

\subsection{Iterated Prisoner's Dilemma}
Imagine two thieves $X$, $Y$ try to rob a bank, but get caught by the police. Unfortunately the police don't have enough evidence to convict the burglars for the robbery, but they can charge them for smaller crimes. They decide to lock the prisoners up separately. The police then gives each of the prisoners a choice to either betray(D) their co-conspirators, giving them up to the authorities or to cooperate(C) with the other prisoner and not give them up. Then here are the possible outcomes:

\begin{itemize}
    \item $X$ and $Y$ both betray the other, then each of them serves two years in prison (or both get a payout of $P$)
\item $X$ betrays $Y$, but $Y$ remains silent, then $X$ won't face jail time, but $Y$ will serve three years in prison (or $X$ gets a payout of $T$ and $Y$ gets a payout of $S$)
\item $X$ remains silent, but $Y$ betrays $X$, then $X$ will serve three years in prison and $Y$ won't face jail time (or $Y$ gets a payout of $T$ and $X$ gets a payout of $S$)

\item If $X$ and $Y$ both remain silent, both serve only one year in prison (or both get a payout of $R$)
\end{itemize}

This premise is formalized in the Prisoners Dilemma, a two player 2x2 game represented by the following matrix:

\begin{equation}
\label{eqn:payoffmat}
\begin{game}{2}{2}
      & $C$     & $D$\\
$C$   & $R,R$  & $S,T$\\
$D$   & $T,S$   & $P,P$
\end{game}
\end{equation}

\noindent The tuple (T, R, P, S) has to be chosen very carefully for the game to actually be interesting.  Values often used include $(5, 3, 1, 0)$ or $(0, -1, -2, -3)$, where the second one is used in the description above. The most interesting behaviours in the game come up when we have the relation $T>R>P>S$. We can extend this setup to the \textbf{Iterated Prisoner's Dilemma} (IPD), where the players get to replay the game multiple times. 

\subsection{Axelrod's Tournament}

In 1980 a professor of University of Michigan, Robert Axelrod, held a tournament for various strategies competing in a game of the IPD. In total he received 14 entries for the competition \cite{kaznatcheev_2015}. He pitted them all against each other in an Iterated Prisoner's Game of length 200 (unknown to the participants). Some of the strategies included "always defect" and "always cooperate",  however, the best strategy turned out to be the Tit-For-Tat strategy reciprocating the opponent's previous move. It was submitted by Anatol Rapoport, a mathematician and psychologist also from the University of Michigan. 

\medskip

This strategy led people to believe some qualities were essential for a good IPD strategy \cite{plato}:

\begin{itemize}
    \item Nice - never the first to defect
    \item Retaliatory - cannot be exploited
    \item Forgiving - willing to cooperate despite being defected on earlier
    \item Clear - predictability makes it easier to cooperate 
\end{itemize}

Commonly a game of multiple strategies put against each other is referred to as an \textbf{Axelrod Tournament}.

\section{Prior Work}

 Research on the IPD has many implications in psychology, game theory, evolutionary biology and many other fields. It is thus not surprising that since Axelrod many other people have gone on to study the IPD. Below we will discuss some of the strategies and their implications mentioned in 2 significant sources. 
 
\subsection{The Calculus of Selfishness}
In his 2009 book \textit{The Calculus of Selfishness} \cite{calcofselfishness} the mathematician Karl Sigmund looks to answer questions about selfishness and cooperation from the perspective of game theory. It is thus natural that a lot of the discussion has implications for the IPD. In this section we will try to study the game of IPD from the point of view of replicator dynamics and see how some strategies evolve over time. We can consider the IPD as we formulated before in \ref{eqn:payoffmat}. Now since the games are iterated, we want a stochastic way to model how long a game lasts. We can thus introduce a variable $\omega \in (0,1)$. Then at each round with probability $\omega$ the game is played again. This can be thought of as a geometric distribution, where we are waiting for a success (the game ending), which has a probability $1-\omega$. Thus the expected game length is $\frac{1}{1-\omega}$. 

\medskip

We will now study the replicator dynamics with 3 strategies. Let these be Cooperators (who always cooperate), Defectors (who always defect), Reciprocators/ Tit-For-Tat (who reciprocate the opponents last move). Then we can introduce a three dimension vector $x$, where $x_1, x_2, x_3$ represents what fraction of the population are Cooperators, Defectors and Reciprocators respectively. Where obviously $x_1+x_2+x_3 = 1$. Then we can model the average payout per round expected for this population using the payoff matrix $A$. To fill in the values of A, we can consider quantities like:

\begin{itemize}
    \item $e_1 \cdot Ae_1 = R$, since if both players cooperate, then we get a payout of R each round
    \item $e_1 \cdot Ae_2 = S$, since if we cooperate against a Defector, we get a payout of S each round
    \item $e_2 \cdot Ae_3 = \frac{T + (\frac{1}{1-\omega}-1)P}{\frac{1}{1-\omega}} = (1-\omega)T + \omega P$, since if we defect against a Reciprocator then in the first round we get a payout of T, but in the subsequent rounds we get a payout of P
    \item etc.
\end{itemize}

hence completing the whole matrix with such considerations, we finally arrive at:

$$A =  \left(\begin{array}{ccc}R & S & R \\ T & P & (1-w) T+w P \\ R & (1-w) S+w P & R\end{array}\right)$$

where we can subtract multiples of $\mathbf{1}$ from the columns of the matrix, because that doesn't change the replicator dynamics. Hence the same replicator dynamics is obtained for:

\begin{align}
\label{eqn:repl}
A =  \left(\begin{array}{ccc}0 & 0 & 0 \\ T-R & P-S & (1-w) T+w P -R\\ 0 & w(P-S) & 0\end{array}\right)
\end{align}

we can clearly see our deliberations will have two cases. When $(1-\omega)T+\omega T > R$ then $e_2$ will be the global Best Response and when $(1-\omega)T+\omega T < R$, we will have $BR(e_1) = e_2, BR(e_2) = e_2, BR(e_3) = \langle e_1, e_3 \rangle$. We can also write out:

\begin{align*}
\begin{cases}
    (Ax)_1 = 0 \\
    (Ax)_2 = x_1 (T-R) + x_2 (P-S) + x_3 ((1-\omega)T + \omega P - R)\\
    (Ax)_3 = \omega (P-S)x_2
    \end{cases}
\end{align*}

now unfortunately at this stage our equations have too many variables and we wouldn't be able to make interesting claims about them, hence we may parametrize our variables by $b, c \in \mathbb{R_+}$, where $b>c$, with $(T, R, P, S) = (b, b-c, 0, -c)$. Then our case $(1-\omega)T+\omega T < R$ reduces to $\omega > \frac{c}{b}$ and this is the case we will consider from now on, since $\omega > \frac{c}{b}$ means the games last longer and these are the games we are interested in. Our previous equations reduce to:

\begin{align*}
\begin{cases}
    (Ax)_1 = 0 \\
    (Ax)_2 = c - \omega b x_3\\
    (Ax)_3 = \omega c x_2
    \end{cases}
\end{align*}

hence we get the indifference lines $Z_{1,2} = \{x_3 = \frac{c}{b}\}$, $Z_{1,3} = \{x_2 = 0\}$, $Z_{2,3} = \{cx_2 + bx_3 = \frac{c}{\omega}\}$. We can expand out the replicator equations to finally get the following phase portrait:

\begin{figure}[H]
\label{fig:repl}
\centering
  \includegraphics[width=0.7\linewidth]{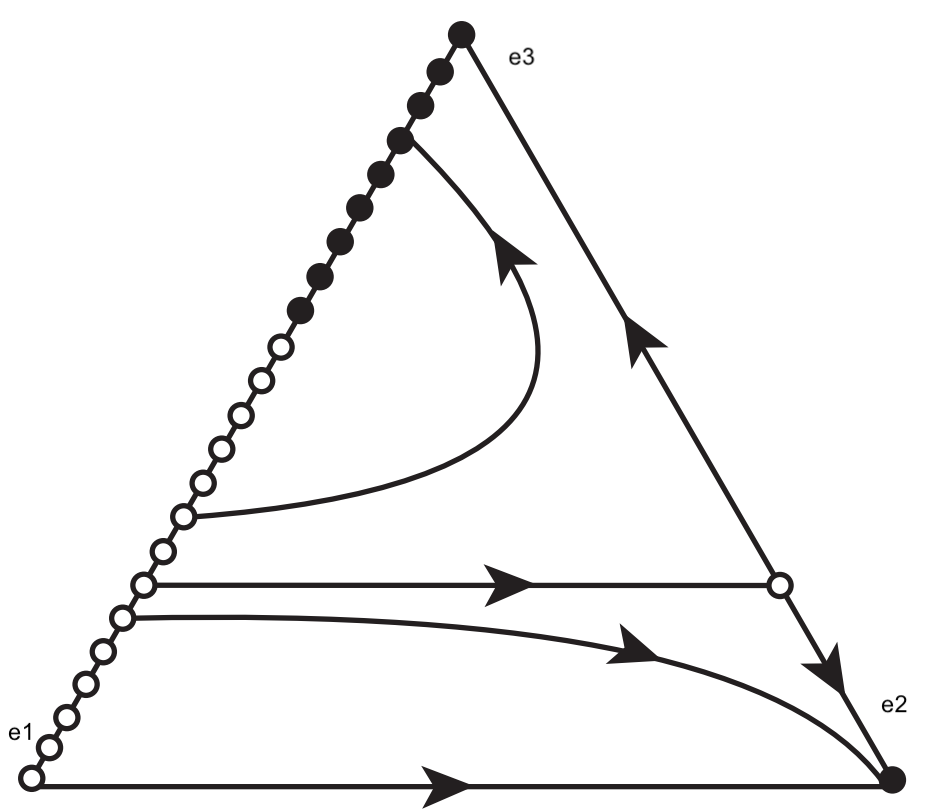}
  \caption{Phase portrait for the system \ref{eqn:repl}, where $(T, R, P, 0) =  (b, b-c, 0, -c)$ and $\omega > \frac{c}{b}$. Filled dots are attracting, whereas white ones are repelling/saddle. From \cite{calcofselfishness}, edited by Anagh Malik.}
\end{figure}

The phase portrait has a very nice interpretation. We can see for populations starting below the line $x_3 = \frac{c}{b}$ the Defectors start dominating and eventually we end up with a population of only those. But this is what we would expect, since always defecting would seem like a "safe bet" - it is a strategy which can't be easily exploited and it is the Prisoners Dilemma's only Nash Equilibrium. However interestingly we can see for populations starting above the line $x_3 = \frac{c}{b}$ the Defectors die out and we are left with a population of more Reciprocators and some Cooperators. This shows the importance of the Reciprocators, who we see have a significance in evolutionary analysis of the IPD. It is good to notice how important our starting position is, as per that we converge to totally different final populations!

\subsection{Zero-determinant Strategies}
In their 2012 paper \cite{press} Press and Dyson claimed to have found a dominant set of strategies they referred to as the \textbf{zero-determinant strategies}. Suppose we are dealing with players $X$ and $Y$ pitted against each other in an IPD game. Intuitively one might think that if Y has a longer memory of previous actions taken by the players in the game than X then it is advantaged. Theoretically we can show that this is not true. 

\begin{prop}
Let $X$, $Y$ be two players playing the IPD game with actions $x, y \in \mathbb{A}=\{C,D\}$. Let the history of previous rounds $H = [H_0, H_1]$, where $H_0$ is the finite history both $X, Y$ have access to  and $H_1$ is the finite history only available to $Y$ i.e. $x$ is independent of $H_1$. We can show that $\forall x, y \in \mathbb{A}$, $\mathbb{E}(P(x,y|H_0, H_1)) = \mathbb{E}(P(x,y|H_0))$, showing that averaged out the probability of an outcome conditioned on the longer history is the same as if conditioned on the shorter history. 
\end{prop}

\begin{proof}

\begin{align*}
\mathbb{E}_H(P(x,y|H_0, H_1))
\stackrel{def}{=}
 \sum_{H_{0}, H_{1}} P\left(x, y \mid H_{0}, H_{1}\right) P\left(H_{0}, H_{1}\right) \stackrel{independence}{=} \\ 
\sum_{H_{0}, H_{1}} P\left(x \mid H_{0}\right) P\left(y \mid H_{0}, H_{1}\right) P\left(H_{0}, H_{1}\right) \stackrel{Bayes}{=} \\ 
\sum_{H_{0}} P\left(x \mid H_{0}\right)\left[\sum_{H_{1}} P\left(y \mid H_{0}, H_{1}\right) P\left(H_{1} \mid H_{0}\right) P\left(H_{0}\right)\right] \stackrel{Bayes}{=} \\
\sum_{H_{0}} P\left(x \mid H_{0}\right)\left[\sum_{H_{1}} P\left(y, H_{1} \mid H_{0}\right)\right] P\left(H_{0}\right) \stackrel{TTP}{=} \\
\sum_{H_{0}} P\left(x \mid H_{0}\right) P\left(y \mid H_{0}\right) P\left(H_{0}\right) \stackrel{def}{=} \\
\mathbb{E}_{H_0}(P(x,y|H_0))
\end{align*}

Where by Bayes we mean the Bayes rule and by TTP we mean the Theorem of Total Probability. Intuitively this means that given a strategy $Y$ with long memory, we can find another short memory strategy, which averaged out performs the same as original strategy against a short memory strategy. Even though this result is theoretical, it has more of an empirical undertone. To make sense of the result it's good to think that we can play a shorter memory strategy, which over a long period will spend the same proportion of time in the state as the longer memory strategy, thus getting the same rewards. 

\end{proof}

This now suggests that given players X, Y, only the length of the shared memory is important, thus if we design a strategy for a player with a short memory, our opponents can't take "advantage" of this. Hence it makes sense for us to design a strategy only conditioned on outcome of the last round. 

\medskip

Let $\Gamma = \{cc, cd, dc, dd\}$ represent the 4 possible outcomes of a game - where c means cooperation and d defection. Then let $\bf{p} \in [0,1]^4$, where $\bf{p}_i$ represents the probability of player X choosing to cooperate, given the outcome of the last round was $xy = \Gamma_i$. We can similarly set a vector $\bf{q} \in [0,1]^4$, where $\bf{q}_i$ represents the conditional probability to cooperate for player Y given the outcome of the previous game was $yx = \Gamma_i$. 

\medskip

We can now construct a $4x4$ matrix $\bf{M}$, such that:

$$\mathbf{M}_{ij} = P((xy)_t = \Gamma_j \mid (xy)_{t-1} = \Gamma_i) $$

where $(xy)_t$ is the outcome of the t-th game. Of course due to time-homogeneity, $\mathbf{M}_{ij}$ isn't dependent on the particular value of t. We can thus use the multiplication rule for probabilities, to see:

\begin{align}
\label{eqn:markovmat}
\mathbf{M(p,q)} = \begin{bmatrix}
    p_1q_1 & p_1(1-q_1) & (1-p_1)q_1 &  (1-p_1)(1-q_1) \\
    p_2q_3 & p_2(1-q_3) & (1-p_2)q_3 &  (1-p_2)(1-q_3) \\
    p_3q_2 & p_3(1-q_2) & (1-p_3)q_2 &  (1-p_3)(1-q_2) \\
    p_4q_4 & p_4(1-q_4) & (1-p_4)q_4 &  (1-p_4)(1-q_4)
\end{bmatrix}
\end{align}

\medskip

Drawing inspiration from stochastic processes, the collection of random variables $\{(XY)_t\}_{t \in \mathbb{N}}$ taking values in $\Gamma$ is a Markov Chain, with transition matrix given by $\mathbf{M(p,q)}$. Suppose now that $\{(XY)_t\}_{t \in \mathbb{N}}$ is an irreducible chain, then it has a unique stationary distribution $\mathbf{v} \in \mathbb{R}^4$ i.e $\mathbf{v}^t \mathbf{M} = \mathbf{v}^t$. The condition for an irreducible chain isn't very restrictive on $\mathbf{p}, \mathbf{q}$. The study of these chains is particularly interesting, since irrespective of the starting distribution the state marginal distribution converges to the stationary distribution\footnote{These are all standard results from the theory of Markov Processes.}. Assuming a quick enough rate of convergence, we can easily calculate the expected rewards for each player per round. Then to compare the payout between two players in a game it would be enough to compare these values, since we can assume each round to have this stationary distribution for the states. With this we can move on to the next big theorem in this section:

\begin{theorem}
\label{thm:doteq}
Given a transition matrix (of an irreducible chain) $\mathbf{M(p,q)}$ as before and its stationary distribution $\mathbf{v}$, for any $f \in \mathbb{R}^4$:

$$\mathbf{v} \cdot f = \frac{D(\mathbf{p}, \mathbf{q}, f)}{D(\mathbf{p}, \mathbf{q}, \mathbf{1})}$$

where for any $b \in \mathbb{R}^4$, we define:

\begin{equation}
\label{eqn:det}
D(\mathbf{p}, \mathbf{q}, b) = \det \begin{bmatrix}
    p_1q_1-1 & p_1 - 1 & q_1-1 &  b_1 \\
    p_2q_3 & p_2 - 1 & q_3 &  b_2 \\
    p_3q_2 & p_3 & q_2-1 &  b_3 \\
    p_4q_4 & p_4 & q_4 &  b_4
\end{bmatrix}
\end{equation}
\end{theorem}

\begin{proof}

Notice $\mathbf{M}\mathbf{1} = \mathbf{1}$, since it is a stochastic matrix. Then the matrix $\mathbf{M'} = \mathbf{M} - \mathbf{I}$ is singular, since $\mathbf{M'}\mathbf{1}=0$. Cramer's Rule \cite{chandra2013notes} tells us that for any matrix A, $\operatorname{adj}(A) A = \det(A) \mathbf{I}$, hence applying Cramer's rule to $\mathbf{M'}$, we see that $\operatorname{adj}(\mathbf{M'}) \mathbf{M'} = 0$, as a singular matrix has a determinant of 0. 

\medskip
Now notice since we assumed the chain is irreducible, $\mathbf{v}$ is the unique vector (up to constant) such that $\mathbf{v}^t \mathbf{M} = \mathbf{v^t}$. Thus $\mathbf{v}$ is also the unique (up to constant) vector such that $\mathbf{v}^t\mathbf{M'} = 0$. So we can deduce that the non-zero rows of $\operatorname{adj}(\mathbf{M'})$ are proportional to $\mathbf{v}$. Thus we can use the non-zero rows of the adjoint to find an expression proportional to the dot product of a vector with $\mathbf{v}$. 

\medskip
If for an nxn matrix A, $\operatorname{adj}(A) = 0$, then we would have $\operatorname{rank}(A) \leq n-2$. But here clearly (as $\mathbf{v}$ is the unique stationary distribution) $\operatorname{rank}(\mathbf{M'}) =3$, hence there is a non-zero row in $\operatorname{adj}(\mathbf{M'})$. Hence wlog we can look at first row of $\operatorname{adj}(\mathbf{M'})$ to get an expression for the dot product with $\mathbf{v}$. To simplify the notation for later, we also introduce the matrix:

$$\mathbf{M''} = \begin{bmatrix}
    p_1q_1 -1 & p_1 -1 & q_1-1 &  (1-p_1)(1-q_1) \\
    p_2q_3 & p_2 -1 & q_3 &  (1-p_2)(1-q_3) \\
    p_3q_2 & p_3 & q_2 -1 &  (1-p_3)(1-q_2) \\
    p_4q_4 & p_4 & q_4 &  (1-p_4)(1-q_4) -1
\end{bmatrix}
$$

then due to simple determinant operations (to obtain $\mathbf{M''}$ we just added the first column of $\mathbf{M'}$ to second and third one) we see:

\begin{align*}
\mathbf{v} \cdot f = K(f_1 \mathbf{M'}_{1,4} + f_2\mathbf{M'}_{2,4} + f_3\mathbf{M'}_{3,4} + f_4\mathbf{M'}_{4,4}) = \\
K(f_1 \mathbf{M''}_{1,4} + f_2\mathbf{M''}_{2,4} + f_3\mathbf{M''}_{3,4} + f_4\mathbf{M''}_{4,4})
\end{align*}

where $K \in \mathbb{R}$ is some proportionality constant and $\mathbf{M}^{(k)}_{i,j}$ is the (i,j)-th minor of $\mathbf{M}^{(k)}$ for $k=1, 2$. We can further expand the minors:

\begin{align*}
\mathbf{v} \cdot f = K(f_1 \mathbf{M''}_{1,4} + f_2\mathbf{M''}_{2,4} + f_3\mathbf{M''}_{3,4} + f_4\mathbf{M''}_{4,4}) = \\
K \left(
f_1 \det \begin{bmatrix}
    p_2q_3 & p_2 -1 & q_3  \\
    p_3q_2 & p_3 & q_2 -1  \\
    p_4q_4 & p_4 & q_4
\end{bmatrix} -
f_2 \det \begin{bmatrix}
    p_1q_1-1 & p_1 -1 & q_1-1  \\
    p_3q_2 & p_3 & q_2 -1 \\
    p_4q_4 & p_4 & q_4
\end{bmatrix} + \right. \\ 
\left. f_3 \det \begin{bmatrix}
    p_1q_1-1 & p_1-1 & q_1-1  \\
    p_2q_3 & p_2-1 & q_3  \\
    p_4q_4 & p_4 & q_4
\end{bmatrix} -
f_4 \det \begin{bmatrix}
    p_1q_1-1 & p_1-1 & q_1-1  \\
    p_2q_3 & p_2-1 & q_3  \\
    p_3q_2 & p_3 & q_2-1  
\end{bmatrix} \right) = \\
K \det \begin{bmatrix}
    p_1q_1-1 & p_1 - 1 & q_1-1 &  f_1 \\
    p_2q_3 & p_2 - 1 & q_3 &  f_2 \\
    p_3q_2 & p_3 & q_2-1 &  f_3 \\
    p_4q_4 & p_4 & q_4 &  f_4
\end{bmatrix} =\\
K D(\mathbf{p}, \mathbf{q}, f)
\end{align*}

We can calculate the exact value of $K$ using the fact that $\mathbf{v}$ is a distribution (values sum to 1):

\begin{align*}
1 = \mathbf{v} \cdot \mathbf{1} =  K D(\mathbf{p}, \mathbf{q}, \mathbf{1})
\implies
K = \frac{1}{D(\mathbf{p}, \mathbf{q}, \mathbf{1})}
\end{align*}

where we can divide out by $D(\mathbf{p}, \mathbf{q}, \mathbf{1})$, because if it was 0 the first equation would be false. Hence combining the two results, we get:

$$\mathbf{v} \cdot f = \frac{D(\mathbf{p}, \mathbf{q}, f)}{D(\mathbf{p}, \mathbf{q}, \mathbf{1})}$$

as required. 

\end{proof}

Now given $\mathbf{p}, \mathbf{q}$ we can calculate the expected score of player $X, Y$ in the case of the marginal state distribution being the stationary distribution of $\mathbf{M}(\mathbf{p}, \mathbf{q})$ (remember that starting with any initial distribution we converge to the stationary one). We calculate this expected score assuming the payoff matrix as in \ref{eqn:payoffmat}, let $s_x, s_y$ be the expected scores given the stationary distribution $\mathbf{v}$ for $X$ and $Y$ respectively, then:

$$
s_x = Rv_1 + Sv_2 + Tv_3 + Pv_4 = v \cdot (R, S, T, P) = \frac{D(\mathbf{p}, \mathbf{q}, S_x)}{D(\mathbf{p}, \mathbf{q}, \mathbf{1})}
$$

$$s_y = Rv_1 + Tv_2 + Sv_3 + Pv_4 = v \cdot (R, T, S, P) = \frac{D(\mathbf{p}, \mathbf{q}, S_y)}{D(\mathbf{p}, \mathbf{q}, \mathbf{1})}
$$

where $S_x = (R, S, T, P)$ and $S_y = (R, T, S, P)$. Due to the linearity of the dot product notice that the operator $D(\mathbf{p}, \mathbf{q}, .)$ is linear in the third argument. We can thus get the expected value of a linear combination of the scores:

$$\alpha s_x + \beta s_y + \gamma = \frac{D(\mathbf{p}, \mathbf{q}, \alpha S_x + \beta S_y + \gamma \mathbf{1})}{D(\mathbf{p}, \mathbf{q}, \mathbf{1})} $$

Now it's important to notice that the player $X$ completely controls the second column of the matrix $\hat{p}= (p_1 - 1, p_2-1, p_3, p_4)$ in \ref{eqn:det} - by setting the vector $\mathbf{p}$. Similarly the player $Y$ completely controls the third row of the same matrix $\hat{q} = (q_1 - 1, q_3, q_2-1, q_4)$ by fixing $\mathbf{q}$. Suppose now that $X$ picks a strategy where $\hat{p} = \alpha S_x + \beta S_y + \gamma \mathbf{1}$ or similarly $Y$ picks a strategy where $\hat{q} = \alpha S_x + \beta S_y + \gamma \mathbf{1}$ for some $\alpha, \beta, \gamma \in \mathbb{R}$\footnote{Such choices for $\hat{p}, \hat{q}$ aren't always possible due to probability constraints, but assume they are for now.}. Then: 

\begin{equation}
\label{eq:zerodet}
\alpha s_x + \beta s_y + \gamma = \frac{D(\mathbf{p}, \mathbf{q}, \alpha S_x + \beta S_y + \gamma \mathbf{1})}{D(\mathbf{p}, \mathbf{q}, \mathbf{1})} = 0
\end{equation}

due to the matrix in the definition of the numerator becoming zero-determinant. Choices for $\mathbf{p}, \mathbf{q}$, which cause this are called zero-determinant (ZD) strategies.

\subsubsection{X sets Y's score}
Notice this way \textbf{$X$ can attempt to set $Y$'s expected score} in the stationary distribution. Let $X$ choose a ZD strategy such that $\hat{p} = \beta S_y + \gamma \mathbf{1}$ i.e:

\begin{equation}
\label{eqn:xsetsscore}
\mathbf{p} = \beta \begin{bmatrix}
           R \\
           T \\
           S \\
           P
         \end{bmatrix} +
         \begin{bmatrix}
           \gamma + 1 \\
           \gamma + 1 \\
           \gamma \\
           \gamma
         \end{bmatrix}
\end{equation}

Then $X$ enforces the linear constraint:
$$\beta s_y + \gamma = 0$$

Previously $\beta, \gamma$ have been fixed, but we want to get rid of this dependency, because they don't always give us actual solutions (due to probability conditions). Hence can first solve the equations \ref{eqn:xsetsscore} for $\beta, \gamma$ in terms of $p_1, p_4$. We get:

$$
\beta = \frac{p_1 - p_4 -1}{R-P}
$$
$$
\gamma = \frac{Rp_4 - Pp_1 +P}{R-P}
$$

Now plugging this into the second and third equation in \ref{eqn:xsetsscore} we can solve for $p_2, p_3$:

\begin{align*}
        p_2 = \beta T + \gamma + 1 = \frac{p_1 - p_4 -1}{R-P} T + \frac{Rp_4 - Pp_1 +P}{R-P} + 1 =\\ \frac{p_1(T-P) - (1+p_4)(T-R)}{R-P}
\end{align*}

\begin{align*}
         p_3 = \beta S + \gamma = \frac{p_1 - p_4 -1}{R-P} S + \frac{Rp_4 - Pp_1 +P}{R-P} =\\ \frac{(1-p_1)(P-S) + p_4(R-S)}{R-P}   
\end{align*}

Of course solutions to these are highly dependent on the exact values we end up choosing for $T, R, P, S$. However having the aforementioned inequality $T>R>P>S$ in mind, we can see that these equations will always have a solution when $p_1$ will be close to 1 and $p_4$ will be close to 0 (to see this you can plug in these values to the equation, then the result follows by continuity). In this case from the equations we can see $p_2$ will be close to 1 and $p_3$ will be close to 0. From equation \ref{eq:zerodet} we can also see that:

\begin{align}
    s_y = -\frac{\gamma}{\beta} = \frac{Rp_4+P(1-p_1)}{p_4 + (1-p_1)}
\end{align}

This means that the score is a weighted average of R and P, with the weights being $1-p_1$ and $p_4$. This forces $s_y$ to be strictly between P and R. Interestingly $X$ gets to set this value by choosing $p_1$ and $p_4$, thus setting Y's expected score in the stationary distribution! The best part about it is that this strategy $\mathbf{p}$ is independent of $\mathbf{q}$, so $X$ can set $Y$'s score independent of the strategy taken by $Y$. We can also interpret this result intuitively. It makes sense for $X$ to only be able to set a score between P and R. $Y$ can guarantee a minimum score of $P$ by always defecting. However $X$ can't guarantee $Y$ a score larger than $R$ (by just being able to control his own actions), since this would require $Y$ to defect.

\begin{example}
We will work out such strategies for the case $(T, R, P, S) = (5, 3, 1, 0)$. In this case the equations for $p_2, p_3, s_y$ reduce to:

\begin{align*}
    p_2 = 2p_1-(1+p_4)\\
    p_3 = \frac{(1-p_1) + 3p_4}{2}\\
    s_y = \frac{3p_4+(1-p_1)}{p_4 + (1-p_1)}
\end{align*}

Looking at the last equality let's first set $p_4$ to 0 (to enforce a low score $s_y$). Then from the first equality we need $\frac{1}{2} \leq p_1 \leq 1$. We can pick $p_1 = \frac{1}{2}$ and $p_1 = \frac{3}{4}$ to get two strategies enforcing $s_y = 1$. These are $\mathbf{p} = (\frac{1}{2}, 0, \frac{1}{4}, 0)$ and $\mathbf{p} = (\frac{3}{4}, \frac{1}{2}, \frac{1}{8}, 0)$ respectively.

\smallskip

We can also pick $s_y$ more moderately setting $p_4 = \frac{1}{2}$. We then have $\frac{3}{4} \leq p_1 \leq 1$. Picking extreme ends we get $\mathbf{p} = (\frac{3}{4}, 0, \frac{7}{8}, \frac{1}{2})$, where $s_y = \frac{7}{3}$ and  $\mathbf{p} = (1, \frac{1}{2}, \frac{3}{4}, \frac{1}{2})$, where $s_y = 3$.

\end{example}

\subsubsection{X sets her own score}

What if \textbf{$X$ tries to fix it's own expected score} in the stationary distribution case? We can attempt to do this by setting $\hat{p} = \alpha S_x + \gamma \mathbf{1}$, to get $\alpha s_x + \gamma = 0$. Unfortunately in this case we don't get much freedom to maneuver $p_1$ and $p_4$. Following the same calculation as before (removing the parameters and conditioning on $p_1$ and $p_4$) we get:

\begin{align}
    p_2 = \frac{(1+p_4)(R-S) - p_1(P-S)}{R-P} \geq 1\\
    p_3 = \frac{-(1-p_1)(T-P) - p_4(T-R)}{R-P} \leq 0
\end{align}

where the first inequality can be seen by plugging in the minimizing values of $p_1 = 1$ and $p_4 = 0$, while the second by plugging in the maximizing values of $p_1 = 1$ and $p_4 = 0$\footnote{Here we again make use of the inequality $T>R>P>S$}. Hence the only strategy we get here is $\mathbf{p} = (1, 1, 0, 0)$ - which irrespective of $Y$'s strategy $\mathbf{q}$, would cause $\{(XY)_t\}_{t \in \mathbb{N}}$ to be an irreducible Markov Chain. Hence due to the probability restrictions $X$ is not able to set it's own score to any desired value by just changing $\mathbf{p}$.

\subsubsection{Extortionate Strategies}
Finally we introduce the most sophisticated strategies, extortionate strategies, where \textbf{X tries to enforce an extortionate share of the payoff} - forcing his share to be larger than P. This can be done by choosing the strategy:

\begin{align}
\label{eqn:extr}
\hat{p} = \Phi((S_x - P\mathbf{1}) - \chi(S_y - P\mathbf{1}))
\end{align}

where $\Phi$ and $\chi$ are constants. This enforces the condition: 

$$\Phi((s_x - P) - \chi(s_y - P)) = 0$$

or more succinctly:

\begin{equation}
\label{eqn:ratio}
    s_x - P = \chi(s_y - P)
\end{equation}

from here we can see that such a condition with the additional assumptions of $s_x > P$ and $\chi \geq 1$ gives $X$ a larger payoff than $Y$ and this also motivates the name given to the parameter $\chi$, the \textit{extortion factor}. We can solve \ref{eqn:extr} for $\mathbf{p}$, by first writing it out in vector form:

\begin{equation}
    \mathbf{p} = \begin{bmatrix}
           1 \\
           1 \\
           0 \\
           0
         \end{bmatrix} +
         \Phi \left( \left(\begin{bmatrix}
           R \\
           S \\
           T \\
           P
         \end{bmatrix} - \begin{bmatrix}
           P \\
           P \\
           P \\
           P
         \end{bmatrix} \right) - \chi \left(
         \begin{bmatrix}
           R \\
           T \\
           S \\
           P
         \end{bmatrix}
         - \begin{bmatrix}
           P \\
           P \\
           P \\
           P
         \end{bmatrix} \right) \right)
\end{equation}

we can rescale this equation, by introducing $\phi = \Phi(P-S)$. Then by further simplification, we finally get:

\begin{align}
\label{eqn:pvalues}
\begin{dcases}
p_{1}=1-\phi(\chi-1) \dfrac{R-P}{P-S} \\
p_{2}=1-\phi\left(1+\chi \dfrac{T-P}{P-S}\right) \\
p_{3}=\phi\left(\chi+\dfrac{T-P}{P-S}\right) \\
\displaystyle
p_{4}=0
\end{dcases} 
\end{align}

Fixing $\chi$, even from \ref{eqn:extr} we see that $\phi$ acts as a scale parameters and thus  by properly adjusting it, the above system of equations always has a solution. We can study the inequalities imposed by the first three of the above equations on $\phi$:

\begin{align*}
\begin{dcases}
0 \leq \phi \leq \dfrac{P-S}{(\chi - 1) (R-P)} \\
0 \leq \phi \leq \dfrac{P-S}{\chi(T-P) + (P-S)} \\
0 \leq \phi \leq \dfrac{P-S}{\chi(P-S) + (T-P)} 
\end{dcases}
\end{align*}

where clearly the second and third inequalities are the most restrictive - unfortunately we can't say which one is more restrictive in general. Also notice that $\phi = 0$, gives only the singular (non irreducible) strategy $\mathbf{p} = (1,1,0,0)$. Hence taking both of these things into consideration, the bound we get for $\phi$ is:

\begin{align}
\label{eqn:phibound}
0 < \phi \leq \min \left\{ \frac{P-S}{\chi(T-P) + (P-S)}, \frac{P-S}{\chi(P-S) + (T-P)}\right\}
\end{align}

The lower bound changes (min is established) depending on whether $T+S>2P$ or $T+S \leq 2P$. We can see this by analyzing when the difference of the numerators is positive (one is greater than the other): 

$$\chi(T-P) + (P-S) - \chi(P-S) - (T-P) = \chi (T+S-2P) - (T+S-2P)$$ 

which is a line with the slope sign dependent on the previously quoted inequality. If $2P > T+S$ then the slope is negative and for all values of $\chi > 1$ the line takes on negative values, thus giving us a different minimum bound (the second element of the minimum operator in \ref{eqn:phibound} is smaller). 

\begin{remark}
\label{eqn:rem}
In the paper the inequality has just been given as:
$$
0 < \phi \leq  \frac{P-S}{\chi(T-P) + (P-S)}$$

however as previously mentioned this is just a necessary inequality, but it is not sufficient i.e. not all values of $\phi$ in this range are allowed. To see this more clearly we can use the values $(T, R, P, S) = (1.5, 1.25, 1, 0)$. These are permitted, since $T>R>P>S$. However we have $2P > T+S$, hence:

\begin{align*}
    \frac{P-S}{\chi(T-P) + (P-S)} = \frac{1}{0.5\chi + 1} > \frac{1}{\chi + 0.5} = \frac{P-S}{\chi(P-S) + (T-P)}
\end{align*}

for $\chi>1$. 
\end{remark}

\medskip
Now notice that due to \ref{eqn:ratio}, both the scores will be maximized at once. Furthermore they strategy $Y$ has to take to maximize the scores is "always cooperate" i.e. $\mathbf{q} = \mathbf{1}$. This is obvious since irrespective of $X$'s action, $X$ always gets a larger payoff if $Y$ cooperates. We can now calculate the score for $X$ in this case, using the solution for $\mathbf{p}$ from \ref{eqn:pvalues}:

\begingroup
\addtolength{\jot}{0.3em}
\begin{align*}
    s_x = \frac{D(\mathbf{p}, \mathbf{q}, S_x)}{D(\mathbf{p}, \mathbf{q}, \mathbf{1})} \leq \frac{D(\mathbf{p}, \mathbf{1}, S_x)}{D(\mathbf{p}, \mathbf{1}, \mathbf{1})} = \\
    \dfrac{\det \begin{bmatrix}
    p_1-1 & p_1 - 1 & 0 &  R \\
    p_2 & p_2 - 1 & 1 &  S \\
    p_3 & p_3 & 0 &  T \\
    p_4 & p_4 & 1 &  P
\end{bmatrix} }{\det \begin{bmatrix}
    p_1-1 & p_1 - 1 & 0 &  1 \\
    p_2 & p_2 - 1 & 1 &  1 \\
    p_3 & p_3 & 0 &  1 \\
    p_4 & p_4 & 1 &  1
\end{bmatrix}} = \\ 
    \dfrac{\det \begin{bmatrix}
    0 & p_1 - 1 & 0 &  R \\
    1 & p_2 - 1 & 1 &  S \\
    0 & p_3 & 0 &  T \\
    0 & p_4 & 1 &  P
\end{bmatrix} }{\det \begin{bmatrix}
    0 & p_1 - 1 & 0 &  1 \\
    1 & p_2 - 1 & 1 &  1 \\
    0 & p_3 & 0 &  1 \\
    0 & p_4 & 1 &  1
\end{bmatrix}} = \\ 
\frac{(1-p_1)T + Rp_3}{(1-p_1) + p_3} = \\
    \frac{P(T-R)+\chi(R(T-S)-P(T-R))}{(T-R)+\chi(R-S)}
\end{align*}
\endgroup

\begin{example}
We can again make this more concrete, going back to our example of $(T, R, P, S) = (5, 3, 1, 0)$. We write out the previously calculated probabilities from equation \ref{eqn:pvalues}:
\begin{align*}
    \mathbf{p}=(1-2 \phi(\chi-1), 1-\phi(4 \chi+1), \phi(\chi+4), 0)
\end{align*}
and the bounds for $\phi$ from inequality \ref{eqn:phibound} become:

\begin{align*}
0 < \phi \leq \min \left\{ \frac{1}{4\chi + 1}, \frac{1}{\chi + 4}\right\} = \frac{1}{4\chi + 1}
\end{align*}

for $\chi>1$. We can also calculate the expected score in the stationary distribution for $X$ assuming $Y$ always cooperates:
\begin{align*}
s_x = \frac{P(T-R)+\chi(R(T-S)-P(T-R))}{(T-R)+\chi(R-S)}
 =\frac{2+13 \chi}{2+3 \chi} = 1 + \frac{10}{\frac{2}{\chi} + 3}
\end{align*}

From the last expression we can see that $X$'s expected score is always greater than 3 - the mutual cooperation score and as $\chi \rightarrow \infty$ the score $s_x$ goes to $\frac{13}{3}$. Using the relation between the expected scores of $X$ and $Y$ from \ref{eqn:ratio}, we can also calculate the score $s_y$, when $Y$ always cooperates:

\begin{align*}
    s_y = \frac{s_x-P}{\chi} + P = \frac{\frac{2+13 \chi}{2+3 \chi} -1}{\chi} + 1 = \frac{12+3 \chi}{2+3 \chi} = 1 + \frac{10}{2+3\chi}
\end{align*}

Here $s_y$ is always less than 3 and as $\chi \rightarrow \infty$ the score $s_y$ goes to $1$.

\smallskip
We can now try to work out concrete strategies. Firstly let's pick an extortion factor of 3 and $\phi = \frac{1}{26}$ the value in the middle of the possible range. Then we get $\mathbf{p} = (\frac{11}{13}, \frac{1}{2}, \frac{7}{26}, 0)$, at best (when $Y$ always cooperates) the expected scores work out to $s_x = \frac{41}{11}$ and $s_y = \frac{29}{11}$. 

\smallskip
We can also try a strategy with $\chi = 2$ and $\phi = \frac{1}{18}$. Then $\mathbf{p} = (\frac{8}{9}, \frac{1}{2}, \frac{1}{3}, 0)$ and in the best case scenario we have $s_x = \frac{7}{2}$ and $s_y = \frac{9}{4}$. 

\smallskip

Finally we can pick a fair strategy i.e. one with $\chi = 1$. Interestingly one of the edge values $\phi = \frac{1}{5}$ gives the Tit-For-Tat strategy i.e. $\mathbf{p} = (1, 0, 1, 0)$, but picking $\phi$ again at the center we get $\phi=\frac{1}{10}$. Then $\mathbf{p} = (1, \frac{1}{2}, \frac{1}{2}, 0)$. 

\end{example}

\subsubsection{Can Y escape X's strategy?}
In most of this section we have assumed $X$ and $Y$ to be playing fixed strategies, thus inducing an (irreducible) Markov Chain leading to the stationary distribution of states. In discovering ZD strategies Press and Dyson have given very clear strategies for $X$ to impose a linear constraint between their expected scores in the stationary distribution. However we could still wonder whether $Y$ could just keep changing its strategy to escape reaching this stationary distribution altogether? This is, however, not possible. 

\medskip 

Suppose $X$ plays some ZD strategy, $\mathbf{p}$. For any strategy $\mathbf{q}$, we define $\mathbf{M(q)} \coloneqq \mathbf{M(p,q)}$, where $\mathbf{M(p,q)}$ is the stochastic matrix defined in \ref{eqn:markovmat}. Now notice for set of strategies (indexed by i) $\{ \mathbf{ q^{(i)} }  \}_{i=1}^N$ for some integer N, we have:

\begin{align*}
        \langle \mathbf{M(q^{(i)})} \rangle_i \coloneqq \frac{1}{N} \sum_{i=1}^{N} \mathbf{M(q^{(i)})} = \\ \frac{1}{N} \sum_{i=1}^{N} \begin{bmatrix}
    p_1q_1^{(i)} & p_1(1-q_1^{(i)}) & (1-p_1)q_1^{(i)} &  (1-p_1)(1-q_1^{(i)}) \\
    p_2q_3^{(i)} & p_2(1-q_3^{(i)}) & (1-p_2)q_3^{(i)} &  (1-p_2)(1-q_3^{(i)}) \\
    p_3q_2^{(i)} & p_3(1-q_2^{(i)}) & (1-p_3)q_2^{(i)} &  (1-p_3)(1-q_2^{(i)}) \\
    p_4q_4^{(i)} & p_4(1-q_4^{(i)}) & (1-p_4)q_4^{(i)} &  (1-p_4)(1-q_4^{(i)})
\end{bmatrix} = \\ \begin{scriptsize}
    \frac{1}{N} \begin{bmatrix}
    p_1(\sum_{i=1}^{N}q_1^{(i)}) & p_1(\sum_{i=1}^{N}(1-q_1^{(i)})) & (1-p_1)(\sum_{i=1}^{N} q_1^{(i)}) &  (1-p_1)(\sum_{i=1}^{N}(1-q_1^{(i)})) \\
    p_2(\sum_{i=1}^{N}q_3^{(i)}) & p_2(\sum_{i=1}^{N}(1-q_3^{(i)})) & (1-p_2)(\sum_{i=1}^{N}q_3^{(i)}) &  (1-p_2)(\sum_{i=1}^{N}(1-q_3^{(i)})) \\
    p_3(\sum_{i=1}^{N}q_2^{(i)}) & p_3(\sum_{i=1}^{N}(1-q_2^{(i)})) & (1-p_3)(\sum_{i=1}^{N}q_2^{(i)}) &  (1-p_3)(\sum_{i=1}^{N}(1-q_2^{(i)})) \\
    p_4(\sum_{i=1}^{N}q_4^{(i)}) & p_4(\sum_{i=1}^{N}(1-q_4^{(i)})) & (1-p_4)(\sum_{i=1}^{N}q_4^{(i)}) &  (1-p_4)(\sum_{i=1}^{N}(1-q_4^{(i)})) \end{bmatrix} \end{scriptsize} = \\
    \begin{bmatrix}
    p_1\langle q^{(i)}_{1} \rangle_i  & p_1(1-\langle q^{(i)}_{1} \rangle_i)  & (1-p_1)\langle q^{(i)}_{1} \rangle_i  &  (1-p_1)(1-\langle q^{(i)}_{1} \rangle_i ) \\
    p_2\langle q^{(i)}_{3} \rangle_i  & p_2(1-\langle q^{(i)}_{3} \rangle_i ) & (1-p_2)\langle q^{(i)}_{3} \rangle_i  &  (1-p_2)(1-\langle q^{(i)}_{3} \rangle_i ) \\
    p_3\langle q^{(i)}_{2} \rangle_i  & p_3(1-\langle q^{(i)}_{2} \rangle_i ) & (1-p_3)\langle q^{(i)}_{2} \rangle_i  &  (1-p_3)(1-\langle q^{(i)}_{2} \rangle_i ) \\
    p_4\langle q^{(i)}_{4} \rangle_i  & p_4(1-\langle q^{(i)}_{4} \rangle_i ) & (1-p_4)\langle q^{(i)}_{4} \rangle_i  &  (1-p_4)(1-\langle q^{(i)}_{4} \rangle_i ) 
\end{bmatrix} = \mathbf{M(\langle q^{(i)} \rangle_i)}
\end{align*}

where $\langle q^{(i)} \rangle_i = \frac{1}{N} \sum_{i=1}^{N} \mathbf{q^{(i)}}$.\footnote{We will use this notation for averages, the subscript shows what we take the average over.} Now suppose we play N (where N is large) rounds of the IPD, where $X$ plays the fixed strategy $\mathbf{p}$ and $Y$ plays the strategy $\mathbf{q}^{(i)}$ at round i, where $i = 1, 2, ... N$. Let $\alpha_i$ denote the state in the i-th round of the game, where $\alpha_i \in \{cc, cd, dc, cc\}$. We can also let $M_{\alpha_i \alpha_{i+1}}(\mathbf{q}^{(i)})$ denote the probability of going from state $\alpha_i$ to $\alpha_{i+1}$ in the i-th round. Then let $N_\beta = \sum_{i=1}^{N} \mathbbm{1}_{\{(XY)_{i} = \beta\}}$ count the number of visits to state $\beta$, then we can calculate the expected value of it:

\begin{align*}
    \mathbbm{E}(N_\beta) &=\sum_{i=1}^{N} M_{\alpha_{i} \beta}\left(\mathbf{q^{(i)}}\right) \\
&=\sum_{\alpha} \sum_{i \mid \alpha} M_{\alpha \beta}\left(\mathbf{q^{(i)} }\right) \\
&=\sum_{\alpha} N_{\alpha}\left\langle M_{\alpha \beta}\left(\mathbf{q^{(i)} }\right)\right\rangle_{i \mid \alpha} \\
&=\sum_{\alpha} N_{\alpha} M_{\alpha \beta}\left(\left\langle \mathbf{q^{(i)} }\right\rangle_{i \mid \alpha}\right)
\end{align*}

where $i \mid \alpha$ are values of i s.t. $\alpha_i = \alpha$. The third equality follows by N being large. We can also define a probability of being in state $\beta$, by $P_\beta = \frac{1}{N} \mathbbm{E}(N_\beta)$. Combining the two results we see that (dividing the previous equation by N):

\begin{equation}
P_{\beta}=\sum_{\alpha} P_{\alpha} M_{\alpha \beta}\left(\left\langle \mathbf{q^{(i)}}\right\rangle_{i \mid \alpha}\right)
\end{equation}

We can thus see that the behaviour of the system is just as if $Y$ played the fixed strategy $\mathbf{q}$, where $\displaystyle q_\alpha = \left\langle q^{(i)}_{\alpha} \right\rangle_{i \mid \alpha}$. Hence we can see that $Y$ can't escape an equilibrium, since its single round strategy doesn't matter in the long run. 

\subsubsection{Discussion}

Before we go on to praise the paper we must quail expectations back a bit. Extortionate strategies will not go on to crush all other strategies. In fact against the Tit-For-Tat player the game would just end up with a defect vs defect scenario eventually, since both strategies surely defect given a $DD$ state previously. However in a tournament an extortionate strategy may go on to extort some players, which the Tit-For-Tat strategy may fail to take advantage of. Similarly using strategies that set $Y$'s score, we might come up against an Defector and then we might end up on the losing side (assuming we are also not always defecting). We must understand these strategies are not invincible and have not just defeated the IPD. 

\medskip
However in their paper Press and Dyson also discuss the possibility of $Y$ being an \textbf{evolutionary player}. An evolutionary player is one, which seeks to maximize its own score $s_y$, by making incremental changes to $\mathbf{q}$, disregarding the strategy employed by $X$ and without attempting to alter his behaviour. This situation is perfectly suited for the \textbf{extortionate strategy}. As we saw previously the best strategy for $Y$ against and extortioner is to always cooperate, hence this is what the evolutionary process will converge to. However, due to the linear relation of the scores imposed by the extortionate strategy $X$ will benefit even more from $Y$ cooperating and thus Press and Dyson's paper has given a very straightforward way of extorting an additional payoff against an evolutionary player. 

\medskip

Press and Dyson's paper really shines in showing more than ever the benefit of having a theory of mind player \cite{PDarticle} i.e. a player who understands how it can influence other players through its strategies.

\medskip

Suppose $X$ tries to extort a theory of mind player $Y$ and as Press and Dyson phrase it "goes to lunch". Then the game becomes an ultimatum game \footnote{The ultimatum game is a two player game, where a proposer splits a pile of money. If the second player doesn't agree to the split both get nothing, otherwise they get to keep their shares.}, as $Y$ can either give in to the extortion and let $X$ receive a larger payoff (by cooperating everytime) or sabotage the game by forcing lower payoffs for both players by also playing sub-optimally. 

\medskip

However, we can also have a situation, where $X$ doesn't "go to lunch" against a theory of mind player $Y$. In this case the players will probably end up negotiating a fair strategy between themselves and there is no point "cheating" since it can easily be punished by the other player. 

\medskip

I personally strongly believe that the Press and Dyson paper has added a completely different dimension to the IPD. It is now very easy to beat evolutionary players and take advantage of some simple strategies. Even though many have been saddened by the implications of the paper "nullifying" the importance of qualities such as niceness etc, the paper has brought about a more profound idea of cooperation for theory of mind players, who can agree to set each others scores and play the game fairly. As Press himself said theory of mind players can now “trust but verify”\cite{edge}.  

\medskip 

That said ZD strategies, despite to their effectiveness against many non-sentient players can't dominate a population in the long run. In their 2013 paper \cite{shortcomingZD} Christoph Adami and Arend Hintze showed that initially ZD tend to dominate a population, but as the population comes to contain more ZD strategies, they tend to perform poorly against each other, pushing their score down in general. The paper shows that any boom in ZD strategies population is short lived. This is why these strategies aren't likely to commonly appear in an evolutionary setting.  

\medskip

Of course at the end of the day a tournament comes down to who participates in it, so we can't ever hope for a universally crushing strategy. All we can hope for is that IPD research pushes us forward and educates us more about cooperation, selfishness and other human qualities. Exactly in this way through stressing the importance of "theory of mind players" and encouraging cooperation Press and Dyson's strategy have done exactly that.  

\section{Produced Strategies}
\label{sec:cust}

Here are the strategies we have come across so far for the example game of $(T, R, P, S) = (5, 3, 1, 0)$:

\begin{enumerate}
  \setcounter{enumi}{-1}
  \label{list:strategies}
    \item Tit-For-Tat: $\mathbf{p} = (1, 0, 1, 0)$
    \item Always Defect: $\mathbf{p} = (0, 0, 0, 0)$
    \item Always Cooperate: $\mathbf{p} = (1, 1, 1, 1)$
    \item X sets Y score to 1: $\mathbf{p} = (\frac{1}{2}, 0, \frac{1}{4}, 0)$
    \item X sets Y score to 1: $\mathbf{p} = (\frac{3}{4}, \frac{1}{2}, \frac{1}{8}, 0)$
    \item X sets Y score to $\frac{7}{3}$: $\mathbf{p} = (\frac{3}{4}, 0, \frac{7}{8}, \frac{1}{2})$
    \item X sets Y score to 3: $\mathbf{p} = (1, \frac{1}{2}, \frac{3}{4}, \frac{1}{2})$
    \item Extortionate strategy with $\chi = 2$ and $\phi = \frac{1}{18}$: $\mathbf{p} = (\frac{8}{9}, \frac{1}{2}, \frac{1}{3}, 0)$
    \item Extortionate strategy with $\chi = 2$ and $\phi = \frac{1}{10}$: $\mathbf{p} = (1, \frac{1}{2}, \frac{1}{2}, 0)$
        \item Extortionate strategy with $\chi = 3$ and $\phi = \frac{1}{26}$: $\mathbf{p} = (\frac{11}{13}, \frac{1}{2}, \frac{7}{26}, 0)$
        \item Custom Strategy
\end{enumerate}

Out of pure interest and curiosity we can also try to devise a strategy of our own. 

\medskip
To motivate it, let's first introduce a different way of interpreting the Tit-For-Tat strategy. The usual way of understanding the strategy is thinking of it as a reciprocator, which just reciprocates the opponents previous move. However we could also think of it as generous predictor. The TFT strategy has a very short memory (just last round) and based on that memory it tries to predict the opponents action, using just frequency analysis. So if you defected in the last round, the TFT strategy predicts you'll defect again, so it also has to defect. Whereas if you cooperated in the last round, it predicts you will cooperate but it is generous so it also cooperates with you. Having this interpretation in mind we can devise a strategy, which tries to learn your one memory cooperation probabilities and tries to copy that strategy. However this time we will use frequency analysis over a longer set of rounds. 

\medskip

Hence our player $X$ will have a non time-homogeneous strategy. Inspired by RL, with some exploitation probability $1-\epsilon$ it will pick the action as per the distribution given by: 
\begin{align}
P(X_t = C | (XY)_{t-1} = xy) = \frac{\sum_{i=1}^{t-1} \mathbbm{1}_{\{Y_i = C \wedge (XY)_{i-1} = xy\}}}{N_{xy}}
\end{align}

where $N_{xy}$ as before counts the number of visits to state $xy$. In simple terms this is the estimate of the probability of player $Y$ cooperating given the previous state, calculated by taking the ratio of the number of times it cooperated after being in that state and the number of times it was in that state. With some exploration probability ($\epsilon$) our strategy will pick a random action. You may ask what exploration, our state space is just of size 4! But notice if you play against a stochastic strategy, then your model of their actions may not be accurate, hence it makes sense to explore the state space to see states and model all probabilities. It is important to notice that it's very easy to get stuck in this game, for example against a TFT strategy we could get into a defect vs defect everytime game but a random action now and then helps get out of this (unfortunately it may also actually also push us into the defect vs defect scenario sometimes). Furthermore the first $K$ rounds for exploration the agent will only be playing random actions, to collect data and be able to model the opponents action probabilities accurately. 
We will now proceed to simulate the games and see how the players perform! 

\medskip

\begin{algorithm}[H]
    \caption{Custom Strategy}

    \textbf{Require:} Exploration coefficient $\epsilon$, Explorations rounds $K$  

    \textbf{Initialize:} Set counter of states N to zeros, Set counter of opponent cooperations after a state, D to zeros. 
    \BlankLine
    \For{$t = 1$ \KwTo $K$}{
    $a \sim \operatorname{Ber}(0.5)$ \\
    Take action a \\
    Update D \\
    Update N
    }
    \For{$t = K+1$ \KwTo $\infty$}{
    $L \sim \operatorname{Ber}(1-\epsilon)$ \\
    \eIf{L=1}
    {
    Get last state $xy$ \\
    $a \sim \operatorname{Ber}(D_{xy}/N_{xy})$ \\
    Take action a \\
    Update D \\
    Update N}
    {
    $a \sim \operatorname{Ber}(0,5)$ \\
    Take action a \\
    Update D \\
    Update N
    }}
\end{algorithm} 

\smallskip

where $\operatorname{Ber(\alpha)}$ is the Bernoulli distribution with rate $\alpha$, if we sample an action from a Bernoulli distribution then a success is equivalent to cooperation. In the edge case where $N_{xy} = 0$, we sample the action from $\operatorname{Ber(0.5)}$.

\section{Tournament Results}

\subsection{Custom Strategy with Demo Strategies}
We can first look at the results of a tournament with the Demo Axelrod Library Strategies and the Custom Strategy, that is strategy 10 in the list (\ref{list:strategies}). The library contains the players: Cooperator (always cooperate), Defector (always defect), Tit-For-Tat, Grudger (always defect if opponent ever defected, otherwise cooperate), Random. We have chosen these opponents since their behaviour is well understood and thus it will make easier to interpret our own agent's behaviour.

\medskip

\begin{figure}[H]
\label{fig:01k20}
\centering
\begin{subfigure}{.5\textwidth}
  \centering
  \includegraphics[width=0.9\linewidth]{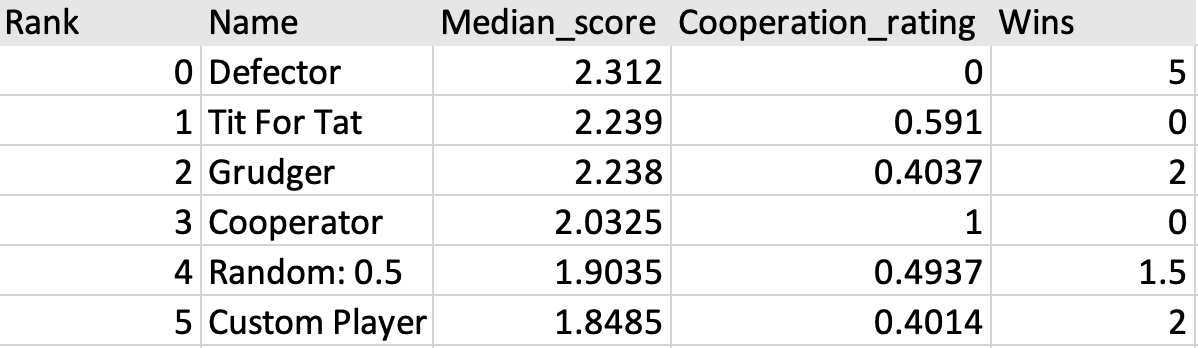}
  \caption{Table summarizing results of the tournament}
\end{subfigure}%
\begin{subfigure}{.5\textwidth}
  \centering
  \includegraphics[width=0.9\linewidth]{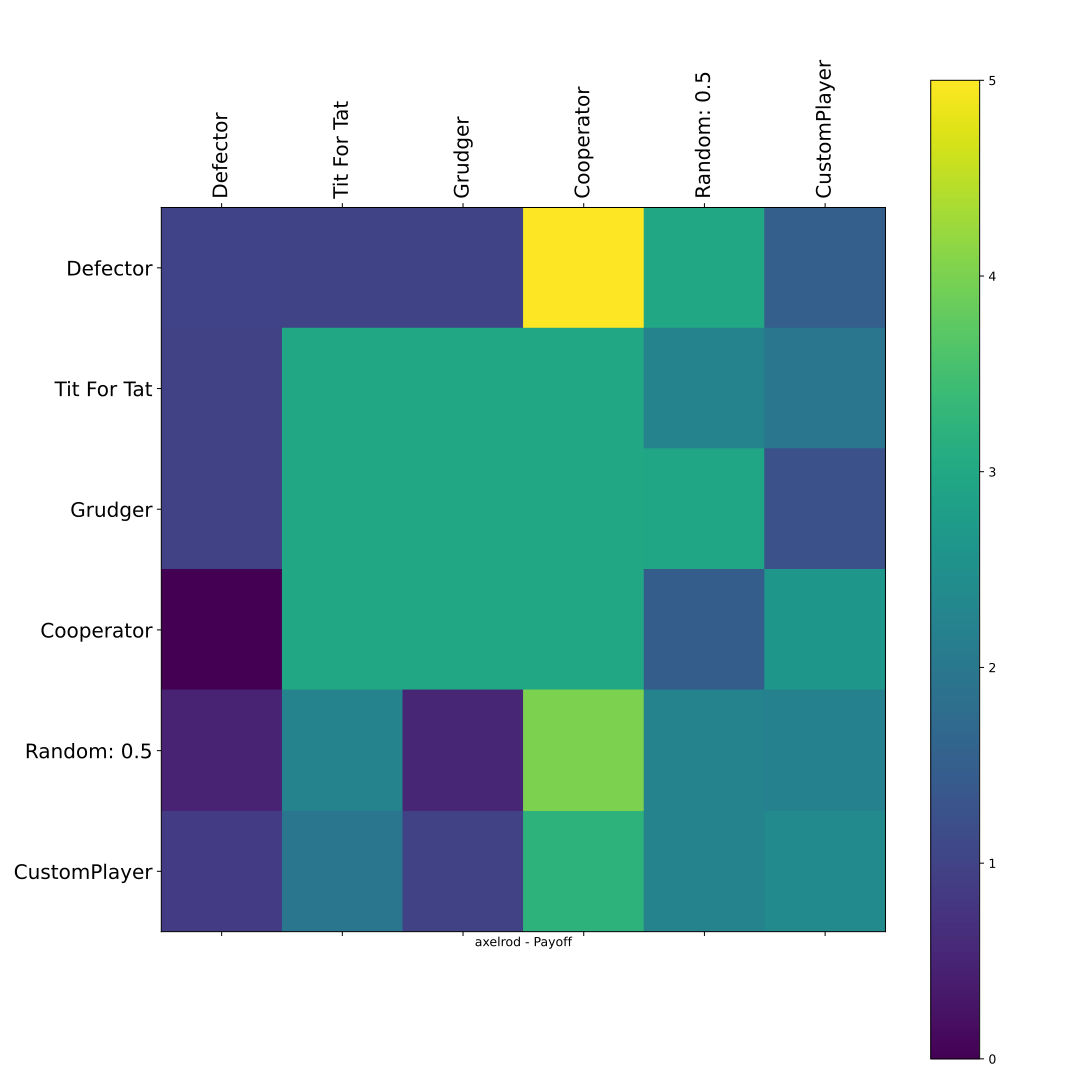}
  \caption{Payoff matrix for particular games}
\end{subfigure}
\caption{Results of experiment with $\epsilon = 0.1$ and $K=20$}
\end{figure}

\begin{figure}[H]
\label{fig:02k20}
\centering
\begin{subfigure}{.5\textwidth}
  \centering
  \includegraphics[width=0.9\linewidth]{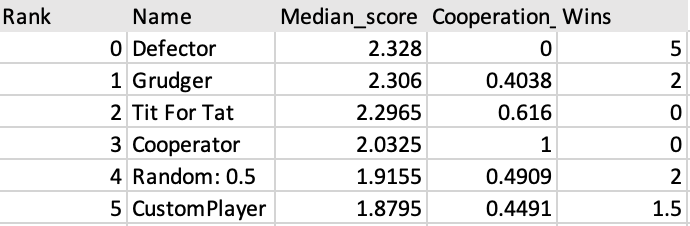}
  \caption{Table summarizing results of the tournament}
\end{subfigure}%
\begin{subfigure}{.5\textwidth}
  \centering
  \includegraphics[width=0.9\linewidth]{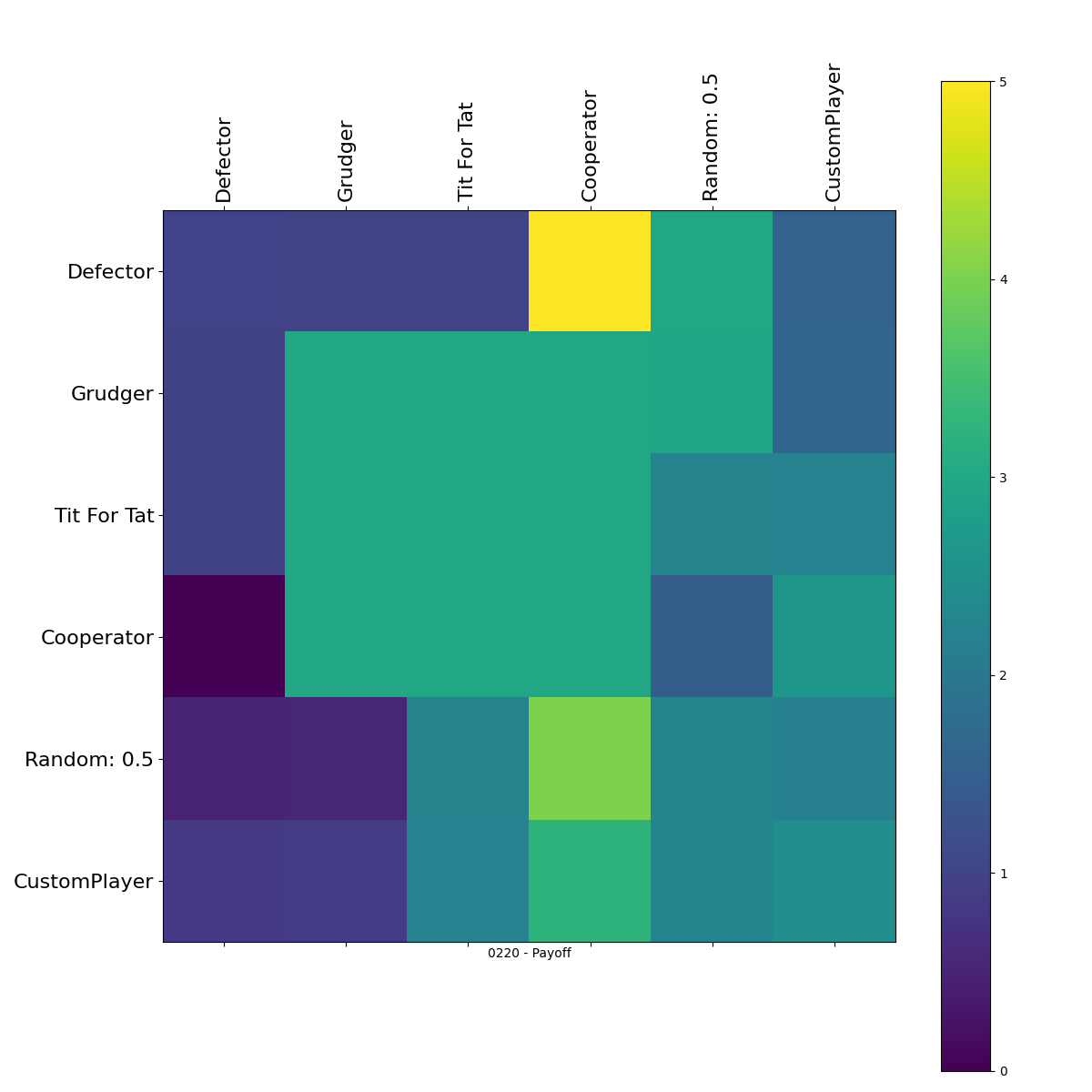}
  \caption{Payoff matrix for particular games}
\end{subfigure}
\caption{Results of experiment with $\epsilon = 0.2$ and $K=20$}
\end{figure}

\begin{figure}[H]
\label{fig:03k20}
\centering
\begin{subfigure}{.5\textwidth}
  \centering
  \includegraphics[width=0.9\linewidth]{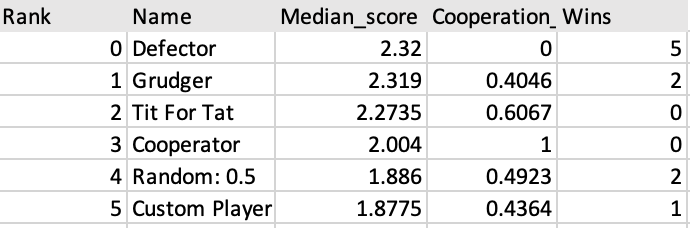}
  \caption{Table summarizing results of the tournament}
\end{subfigure}%
\begin{subfigure}{.5\textwidth}
  \centering
  \includegraphics[width=0.9\linewidth]{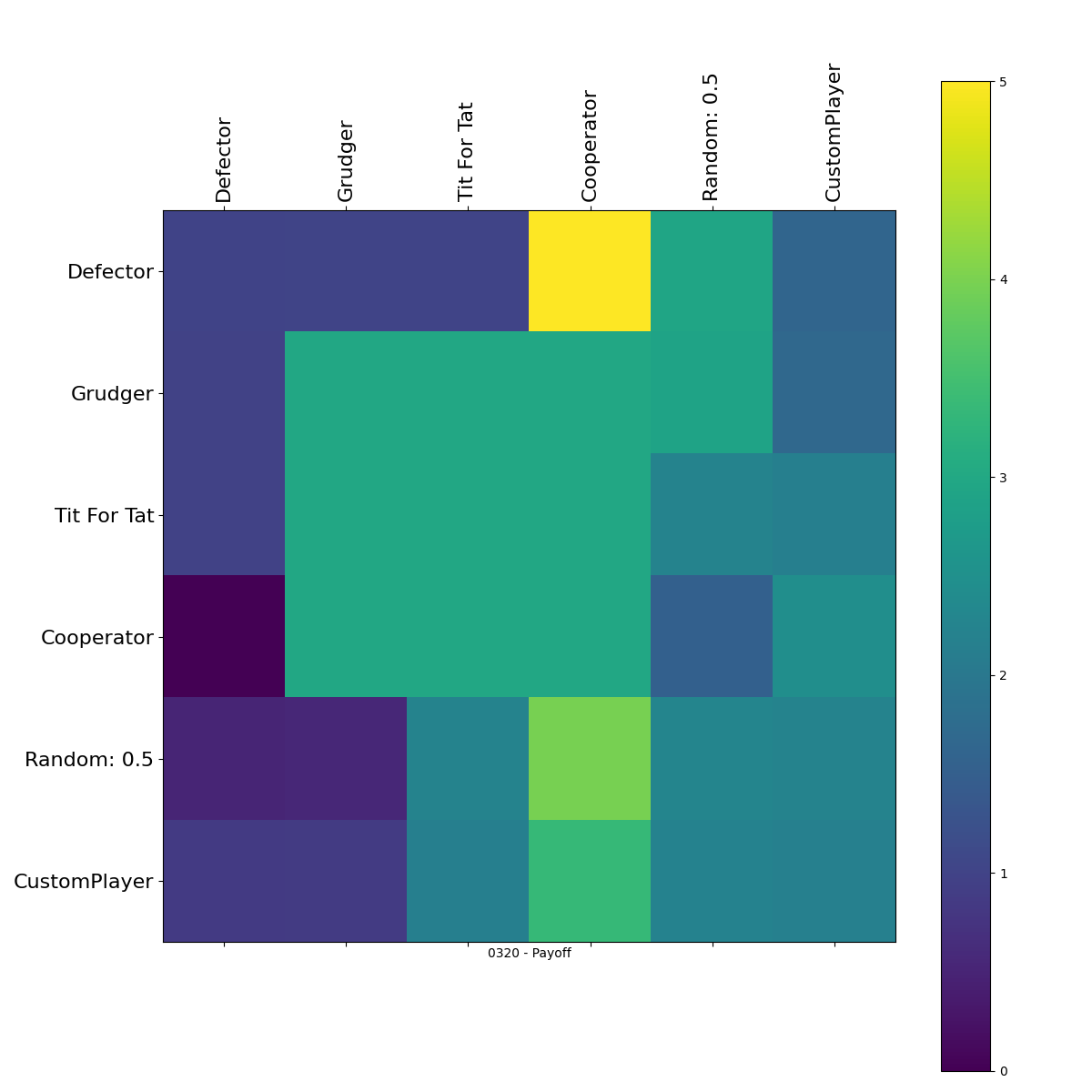}
  \caption{Payoff matrix for particular games}
\end{subfigure}
\caption{Results of experiment with $\epsilon = 0.3$ and $K=20$}
\end{figure}

\begin{figure}[H]
\label{fig:01k40}
\centering
\begin{subfigure}{.5\textwidth}
  \centering
  \includegraphics[width=0.9\linewidth]{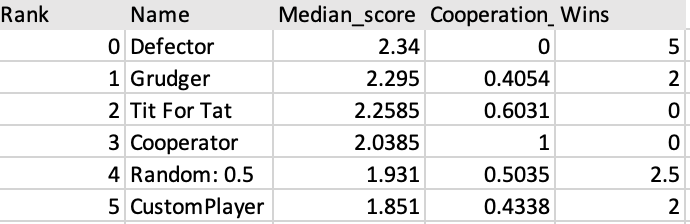}
  \caption{Table summarizing results of the tournament}
\end{subfigure}%
\begin{subfigure}{.5\textwidth}
  \centering
  \includegraphics[width=0.9\linewidth]{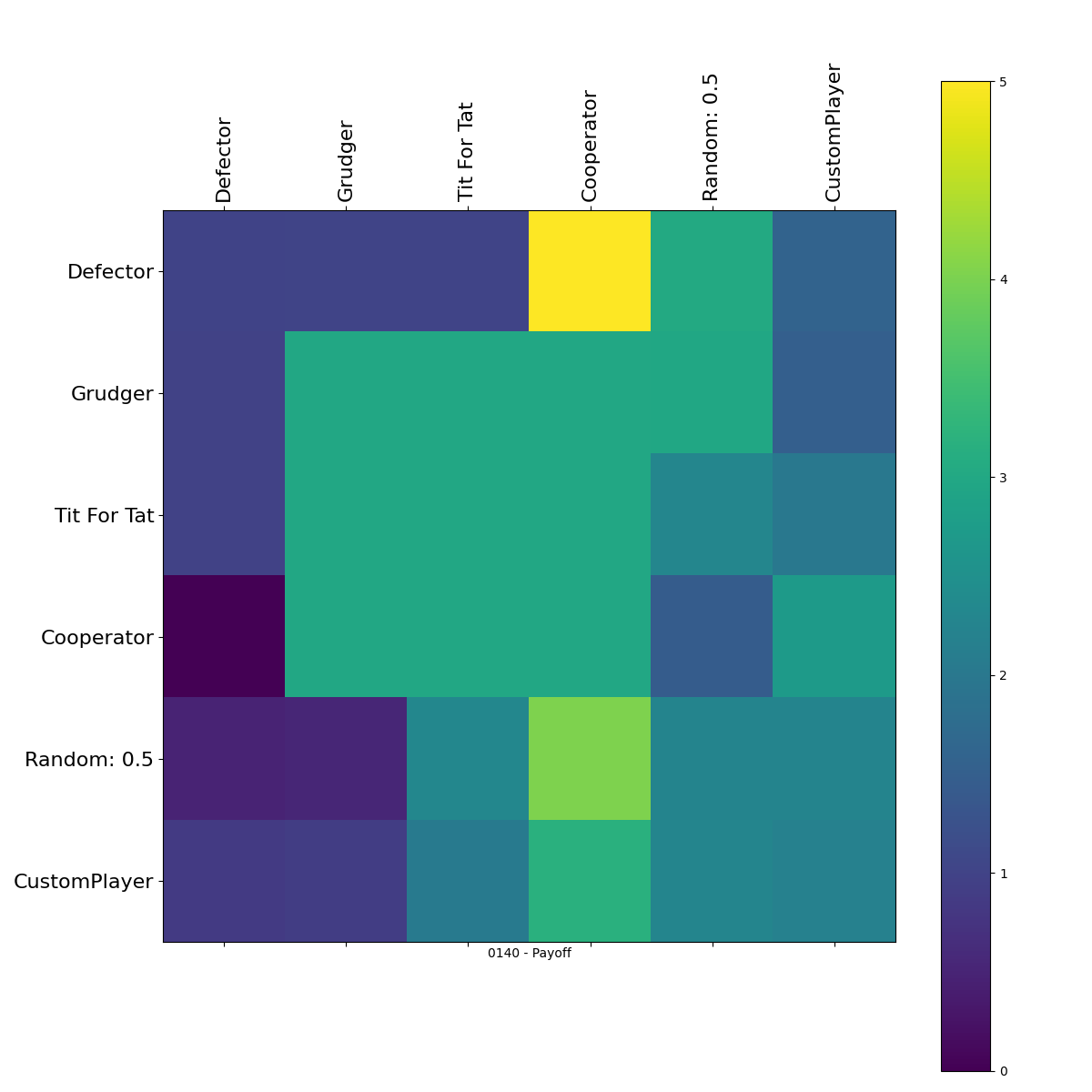}
  \caption{Payoff matrix for particular games}
\end{subfigure}
\caption{Results of experiment with $\epsilon = 0.1$ and $K=40$}
\end{figure}

\begin{figure}[H]
\centering
\begin{subfigure}{.5\textwidth}
  \centering
  \includegraphics[width=0.9\linewidth]{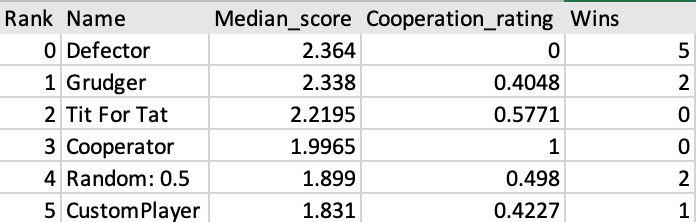}
  \caption{Table summarizing results of the tournament}
\end{subfigure}%
\begin{subfigure}{.5\textwidth}
  \centering
  \includegraphics[width=0.9\linewidth]{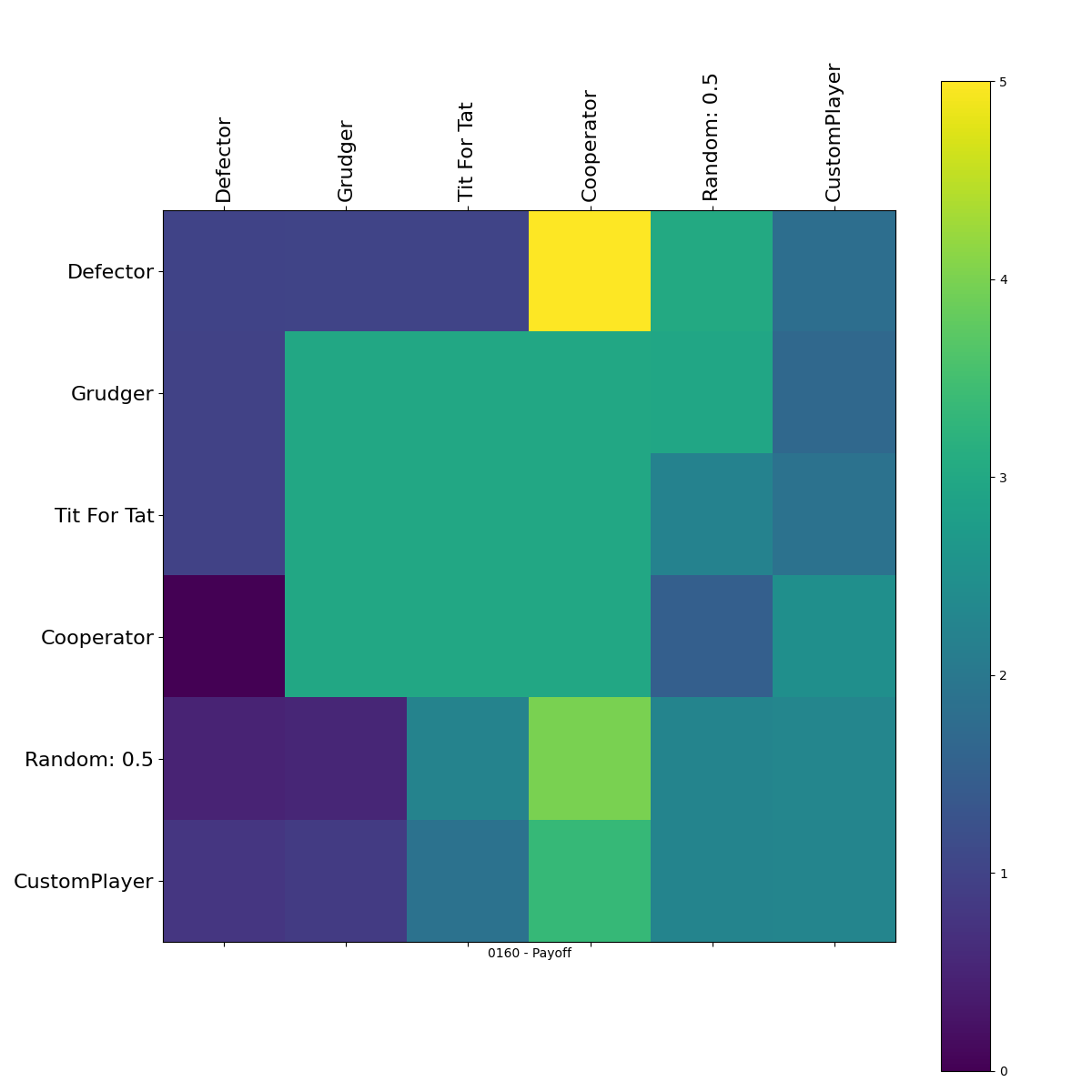}
  \caption{Payoff matrix for particular games}
\end{subfigure}
\caption{Results of experiment with $\epsilon = 0.1$ and $K=60$}
\end{figure}

\begin{figure}[H]
\centering
\begin{subfigure}{.5\textwidth}
  \centering
  \includegraphics[width=0.9\linewidth]{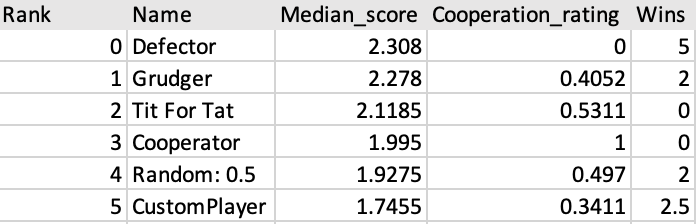}
  \caption{Table summarizing results of the tournament}
\end{subfigure}%
\begin{subfigure}{.5\textwidth}
  \centering
  \includegraphics[width=0.9\linewidth]{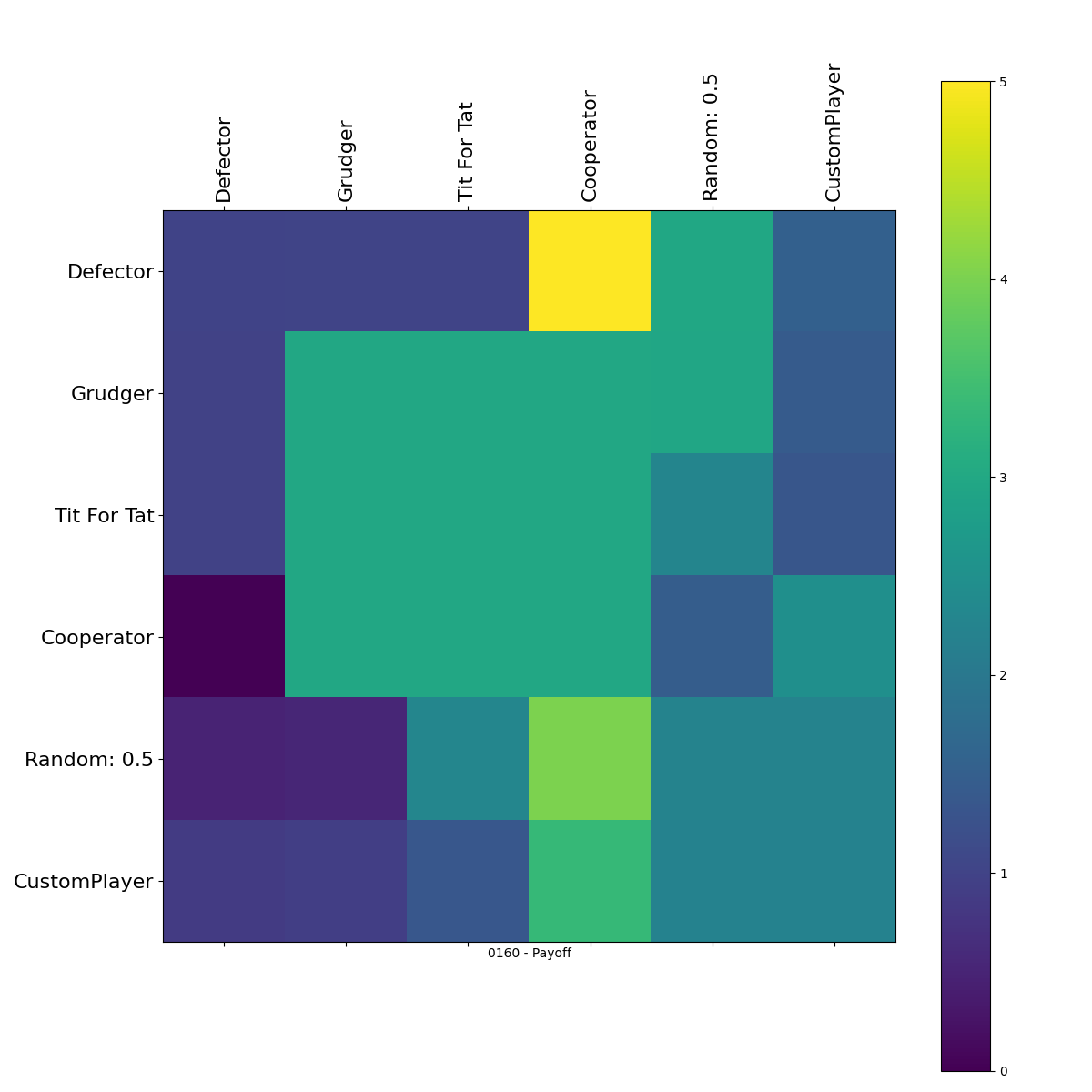}
  \caption{Payoff matrix for particular games}
\end{subfigure}
\caption{Results of experiment with $\epsilon = 0$ and $K=60$}
\end{figure}

\begin{figure}[H]
\centering
\begin{subfigure}{.5\textwidth}
  \centering
  \includegraphics[width=0.9\linewidth]{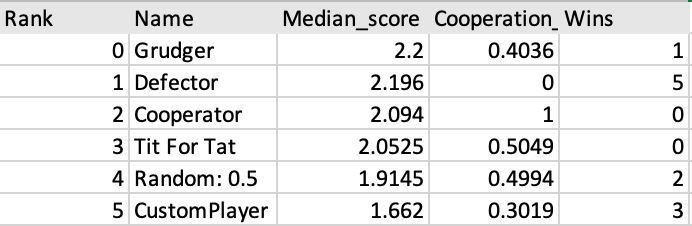}
  \caption{Table summarizing results of the tournament}
\end{subfigure}%
\begin{subfigure}{.5\textwidth}
  \centering
  \includegraphics[width=0.9\linewidth]{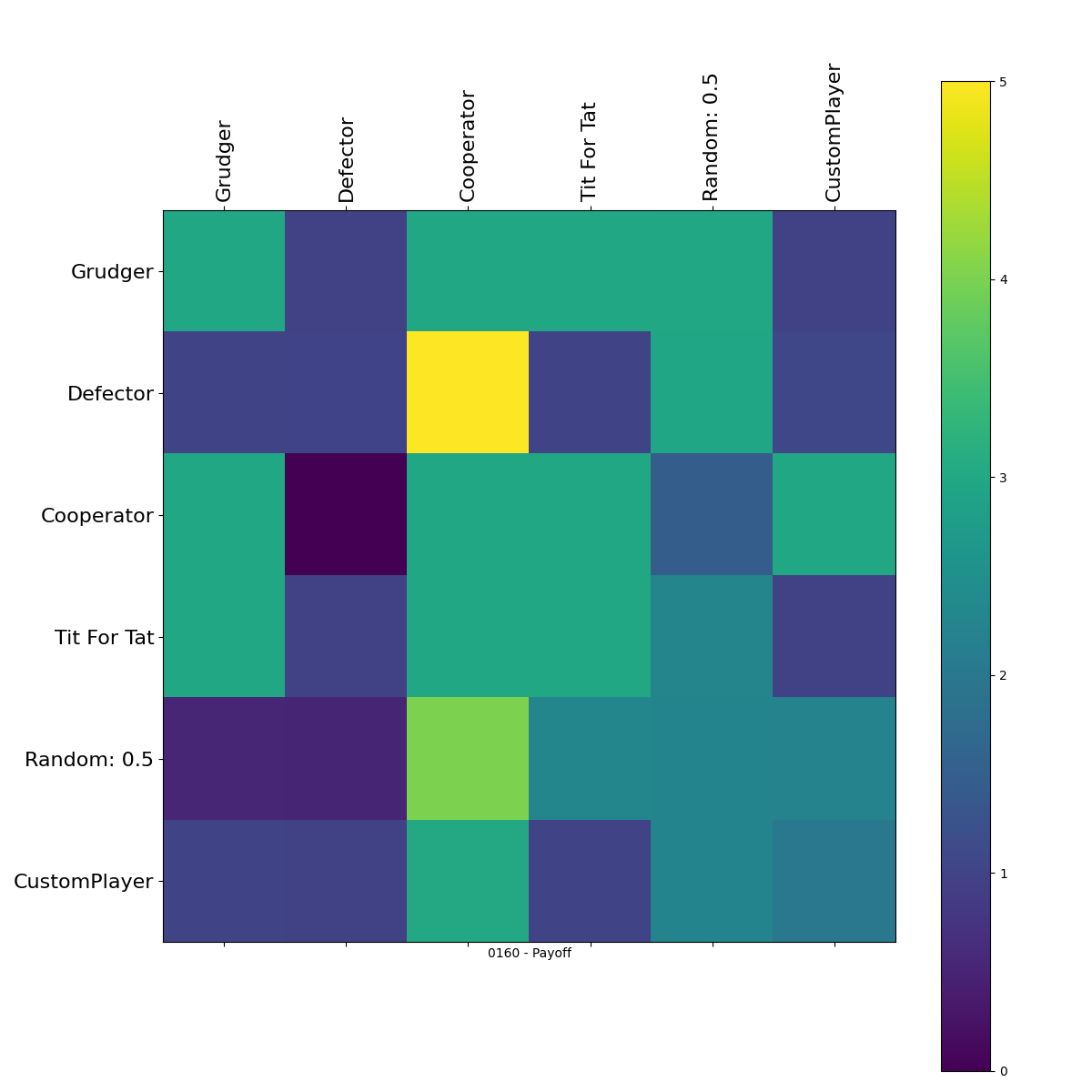}
  \caption{Payoff matrix for particular games}
\end{subfigure}
\caption{Results of experiment with $\epsilon = 0$ and $K=0$}
\end{figure}

\medskip

In the figures above we present the outcomes of the tournaments run with the Custom Player with different values of $\epsilon$ and $K$. 

\medskip
Unfortunately our player doesn't tend to perform as well as we would have liked against this set of opponents - in fact it loses all the tournaments. There's probably many reasons for that we could give, but before that let's study some effects of our parameters $\epsilon$ and $K$ on the outcomes of the tournaments. It's important to mention that we make use of random actions if we don't have any knowledge about the state, and that for the first $K$ rounds we also make use of random actions, thus the results are stochastic and the exact values are only roughly indicative of any behaviours. 

\medskip

We can first estimate the effect of $\epsilon$ on the tournaments. As we can see the strategies with higher $\epsilon$ tend to score more points - this can be better seen in how the Median Score changes as per the value of $\epsilon$. Clearly the Median Score of the Random player is always higher than our Custom Player in these tournaments and increasing our $\epsilon$ value just brings our player closer to the Random player, hence the results are justified. We must however ask, why does exploitation not work well? From the Summary Tables, we can see that the higher $\epsilon$ valued players cooperate more, this is especially visible for the experiments when $\epsilon=0$, where we only cooperate roughly 3 out of 10 times. Cooperating more often works particularly well against the Tit-For-Tat player, as can be also seen in the Pay-Off matrices - the payoff matrices suggest a higher payoff for higher $\epsilon$ players against the TFT strategy. As previously hypothesized the exploration factor thus helps get out of the defect vs defect loop. $\epsilon$ was also in theory supposed to help us better model the opponent cooperation distribution in rarely visited states, but this is a non-issue against the given opponents, since they have deterministic actions. 

\medskip

We can now consider different values of $K$ i.e. the initial period for which we just explore. We can see that the value doesn't drastically change our Mean Score. More explicitly for a fixed value of $\epsilon=0.1$, considering the Mean Score, there is a slight increase when we go from $K=20$ to $K=40$, but a larger decrease from $K=40$ to $K=60$. One might ask why is this? We have already sort of addressed this in the previous paragraph, the value $K$ was supposed to give us a period of exploration, to discover the opponents cooperation distribution, but here as mentioned are deterministic, thus we in fact only need $K=0$ (a value of $K=0$ means we will just start exploiting from the beginning and if we come up against an unknown state, we choose a random action as specified previously). To do what we were planning to do with this strategy, we only need to come across a state once to be able to model the opponents action distribution. Hence in this case a higher $K$ value just gives us more random actions at the beginning without adding anything for our strategy towards the later part. Any change of the Mean Score would only be due to picking random actions for a different number of rounds.

\medskip

\subsubsection{Conclusions}

Finally we can move on to why our strategy hasn't performed well. There is more then one answer to this question, so firstly we will try to bring up some problems with the strategy and towards the end we will suggest changes. 

\smallskip

There is one obvious answer as to why we performed badly - we didn't have the right opponents! This strategy is bad against this set of opponents.  There is a clear reason why, we just don't cooperate as much. We can consider the actions we take against all opponents. With this strategy we would start off acting randomly, which would put off the Grudger - hence leading to him and us both defecting everytime. With this strategy since we are copying the opponent, we couldn't exploit the Cooperator and similarly we got into a defect vs defect game against the Defector. The only real room to maneuver was against the TFT and random player, with whom our average payoffs tent do oscillate around 2 and 3. Against the Random player, assuming we learn to mimic them perfectly, all the states become equally likely, thus the expected payoff per round is $\frac{1}{4}(5+3+1) = 2$, while defecting each time would give us an expected payoff of 3, hence it would be worth defecting. These are small things our strategy fails to take note of and hence we can see that our behaviour leads to sub-optimal payoff for ourselves. However it is worth noting that because we end up copying the opponent is it also not easy to score points off of us - this is also represented in all the payoff matrices. For a similar reason we end up drawing/winning a lot of games - we don't get a lot of points, but we also don't let the opponent get a lot of points! 

\smallskip
I believe the strategy would work well will a more diverse set of opponents, aren't deterministic. But even then more importantly I think this strategy would perform better in games with a lot longer time scales. As mentioned previously we would be missing out a lot on early rewards when setting a high value of K, however the reward lost would be little if the total time horizon was 1000 or even 10000 (as opposed to the current 200). Perhaps I could suggest in a 10000 length game it would be interesting to separate out a long exploration phase, but have no exploration in the exploitation phase i.e. $\epsilon=0$, $K>200$. This way we could after learning the model start from "afresh" and just play the same was as our opponent. This way we also don't have to worry about unexplored states, since we should never have to hit them. However this analysis is left for the next time! This would also be a good strategy to play against a theory of mind player, who would quickly realize that to actually not lose a lot of potential payoffs, he would have to start cooperating, thus changing our cooperation probability distribution. But as mentioned previously for all these points a longer time horizon would be useful.

\subsection{Custom Strategy with ZD Strategies}

We can now run a tournament with all the strategies in out table \ref{list:strategies}. These are stochastic, so we would expect our Custom Strategy to just be a better version of the Tit-For-Tat player. Below are the results:

\begin{figure}[H]
\centering
\begin{subfigure}[b]{1\textwidth}
   \includegraphics[width=1.5\linewidth, center]{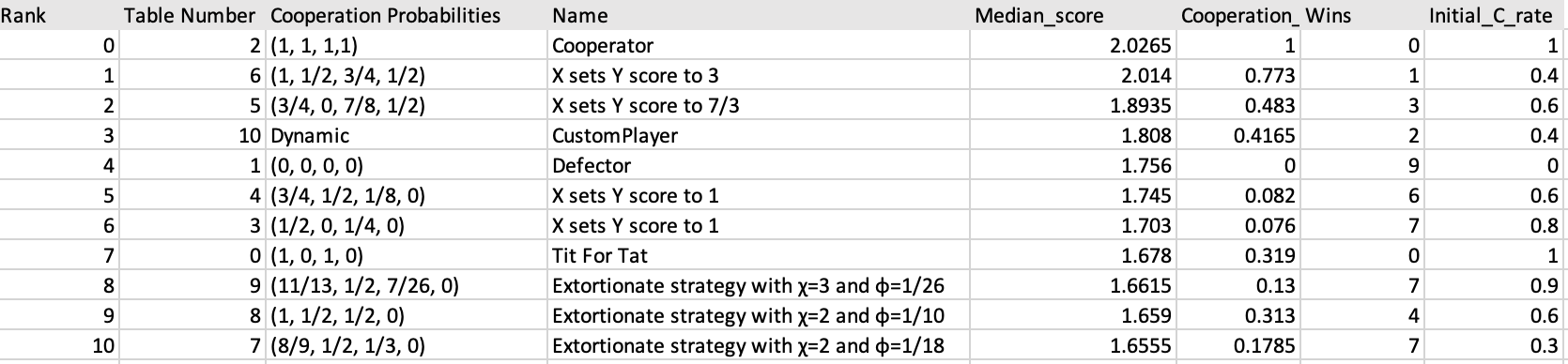}
   \caption{Table summarizing results of the tournament, table number same as in \ref{list:strategies}}
   \label{fig:Ng1} 
\end{subfigure}

\begin{subfigure}[b]{1\textwidth}
   \includegraphics[width=0.8\linewidth, center]{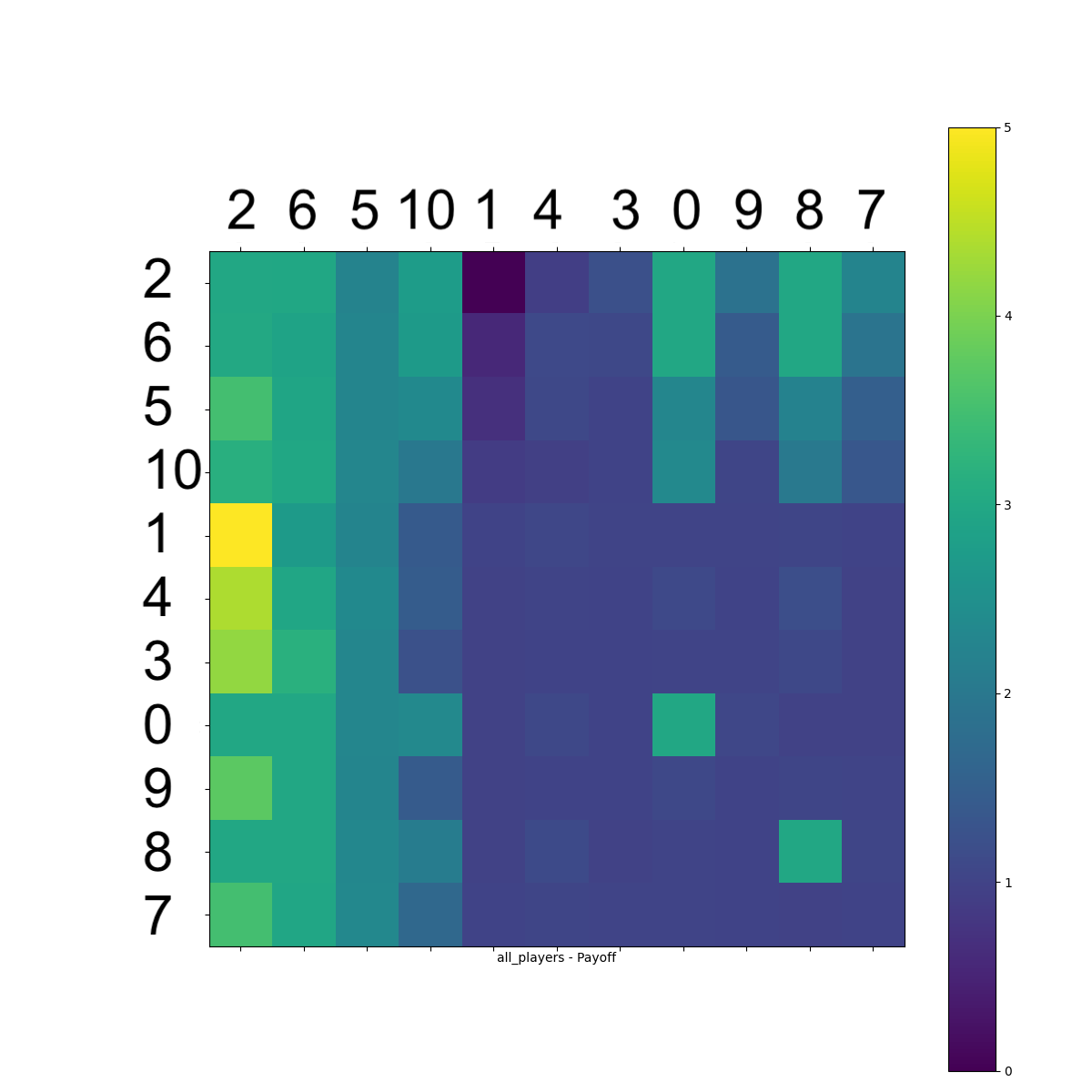}
   \caption{Payoff matrix for particular games, indexed by Table Number from \ref{fig:Ng1}}
   \label{fig:Ng2}
\end{subfigure}
\caption{Tournament results with all strategies we found. Custom strategy used with $\epsilon=0.1$ and $K=20$}
\end{figure}

We can see some very interesting results above, it is especially fascinating to see the ZD strategies at work, for example the strategy where $X$ sets $Y$'s score to $\frac{7}{3}$ especially stands out in \ref{fig:Ng2}. This tournament simulates a lot of strategies, hence we will mostly refer to the strategies by their table numbers to avoid ambiguity, for example the Cooperator is strategy number 2.

\medskip

 The Median Scores are very close for most of the strategies, except the Cooperator who really runs away with it. We can see most of the strategies included have a very high inclination to cooperate after a $CC$ state, thus it is not surprising that the Cooperator's Median Score is so high. 

\medskip

Next in the rankings come the strategies come other generous players. 3 other players have a Median Score above 1.8. This includes the generous ZD strategies setting the score of the opponent to 3 and $\frac{7}{3}$ respectively and our Custom Strategy. These again have an inclination to cooperate and, thus generally end up in $CC$ states many times. From the payoff matrix, we can see that the top 4 strategies have especially benefited by scoring high numbers against the bottom 4 strategies, that is the Tit-For-Tat player and the 3 extortionate strategies. 

\smallskip

On the other hand the center of the table is filled with the most "negative" strategies. That is the Defector and the two strategies setting the opponents score to 1. They score well against the top 3 strategies (mostly exploiting them with $CD$ rounds), but end up having a lot of $DD$ states against the other players, thus pushing their score down. 

\smallskip

I was surprised to see Tit-For-Tat towards the bottom of the table, but on second thought it is not that surprising since once it falls into the defect loop against the players (apart from the top 4) it stays there. It would be interesting to see how Generous Tit-For-Tat would have performed, it seems to be very similar to the strategies at the top. 

\smallskip

It's also interesting to see how the Extortionate strategies perform. They are clearly enforcing their desired outcome against the cooperator, but due to the strategies being non evolutionary in general they tend do get an overall low score against most and end up taking the last places in the table. We can also see that Extortionate strategies tend to perform badly against themselves, this is due to an extortion unbalance. Suppose we have two extortionate strategies $X$ and $Y$ with different factors $\chi = \chi_x$ and $\chi = \chi_y$ respectively. Then:

\begin{align*}
    \begin{cases}
    s_x - P = \chi_x(s_y - P) \\
    \chi_y(s_x - P) = \chi(s_y - P)
\end{cases} \implies
\begin{cases}
s_x = P \\
s_y = P
\end{cases}
\end{align*}

hence they both end up with the same payoff as if the both defected each time. 

\smallskip

What's very interesting is that the top 4 strategies all tend to have a positive probability of cooperating after a $DD$ state, this surely allowed them to escape the $DD$ loop against many of the bottom placed players. This is what the remaining strategies have failed to do, hence why they occupy lower positions. It's also very interesting to see how the number of games won plays absolutely no role. The player in the first position won no games and the player in the last position won 7 games! Finally it's important to mention that our one-memory players cooperate randomly in the first round (this can be seen by the initial cooperation column in \ref{fig:Ng1}), hence the results are stochastic and might have been quite different by just changing this small feature. 

\medskip

Our Custom Strategy as mentioned has performed quite well. It really shines against more stochastic opponents, learning their cooperation probabilities and performing like a better version of the Tit-For-Tat player as intended!

\newpage

\bibliographystyle{plain}
\bibliography{bibliography.bib}

\end{document}